\documentclass[a4paper]{scrartcl}
\pdfoutput=1
\usepackage[utf8]{inputenc}
\usepackage[T1]{fontenc}
\usepackage{lmodern}
\usepackage[sc]{mathpazo}
\usepackage{amsmath,amssymb,amsthm}
\usepackage{enumitem}
\usepackage[british]{babel}
\usepackage{microtype} 
\usepackage{authblk}
\usepackage{hyperref}

\usepackage[numbers,sort&compress]{natbib}
\newcommand{\D}{\mathrm{d}}
\newcommand{\DD}{\mathrm{D}}

\newtheorem{thm}{Theorem}[section]
\newenvironment{thmqed}{\pushQED{\qed}\begin{thm}}{\popQED\end{thm}} 
\newtheorem{lem}[thm]{Lemma}
\newenvironment{lemqed}{\pushQED{\qed}\begin{lem}}{\popQED\end{lem}}
\newtheorem{prop}[thm]{Proposition}

\theoremstyle{definition}
\newtheorem{defn}[thm]{Definition}
\newtheorem{exmp}[thm]{Example}

\numberwithin{equation}{section}

\title{Classical perspectives on the Newton--Wigner position observable}

\author[1,a]{Philip K. Schwartz}
\author[1,2,b]{Domenico Giulini}
\affil[1]{Institute for Theoretical Physics,
	Leibniz University Hannover, \par
	Appelstra{\ss}e 2, 30167 Hannover, Germany}
\affil[2]{Center of Applied Space Technology and Microgravity,
	University of Bremen, \par
	Am Fallturm 1, 28359 Bremen, Germany}
\affil[a]{\normalfont\texttt{\href{mailto:philip.schwartz@itp.uni-hannover.de}{philip.schwartz@itp.uni-hannover.de}}}
\affil[b]{\normalfont\texttt{\href{mailto:giulini@itp.uni-hannover.de}{giulini@itp.uni-hannover.de}}}

\date{}

\begin{document}
\maketitle

\begin{abstract}
	\noindent
	This paper deals with the Newton--Wigner position 
	observable for Poincaré-invariant \emph{classical} 
	systems. We prove an existence and uniqueness 
	theorem for elementary systems that parallels the 
	well-known Newton--Wigner theorem in the quantum 
	context. We also discuss and justify the geometric 
	interpretation of the Newton--Wigner position as 
	`centre of spin', already proposed by Fleming in 
	1965 again in the quantum context.
\end{abstract}

\section{Introduction}

Even though we shall in this paper exclusively deal 
with \emph{classical} (i.e.\ non-quantum) aspects of 
the Newton--Wigner position observable, we wish to 
start with a brief discussion of its historic 
origin, which is based in the early history of 
relativistic quantum field theory (RQFT). After that 
we will remark on its classical importance and 
give an outline of this paper.  

The conceptual problem of how to properly `localise' a 
physical system `in space' has a very long history, the 
roots of which extend to pre-Newtonian times. Newtonian 
concepts of space, time, and point particles allowed 
for sufficiently useful localisation schemes, either in 
terms of the position of the particle itself if an 
elementary (i.e.\ indecomposable) systems is considered, 
or in terms of weighted convex sums of instantaneous 
particle positions for systems composed of many particles, 
like, e.g., the centre of mass. These concepts satisfy 
the expected covariance properties under spatial 
translations and rotations and readily translate to 
ordinary, Galilei-invariant Quantum Mechanics, where 
concepts like `position operators' and the associated 
projection operators for positions within any measurable 
subset of space can be defined, again fulfilling the 
expected transformation rules under spatial motions.

However, serious difficulties with naive localisation 
concepts arose in attempts to combine Quantum Mechanics 
with Special Relativity. For example, as already 
observed in 1928 by Breit \cite{Breit:1928} and again 
in 1930--31 by Schrödinger 
\cite{Schroedinger:1930,Schroedinger:1931}, a naive 
concept of `position' for the Dirac equation leads to 
unexpected and apparently paradoxical results, like the 
infamous `Zitterbewegung'. It soon became clear that 
naive translations of concepts familiar from 
non-relativistic Quantum Mechanics did not 
result in satisfactory results. This had to do with 
the fact that spatially localised wave functions (e.g.\ 
those of compact spatial support) necessarily contained 
negative-energy modes in their Fourier decomposition 
and that negative-energy modes would necessarily be 
introduced if a `naive position operator' (like 
multiplying the wave function with the position 
coordinate) were applied to a positive-energy state. 
Physically this could be seen as an inevitable result 
of pair production that sets in once the bounds on 
localisation come close to the Compton wavelength. 
Would that argument put an end to any further attempt 
to define localised states in a relativistic context?

This question was analysed and answered in the 
negative in 1949 by Newton and Wigner 
\cite{Newton.Wigner:1949}. Their method was to 
write down axioms for what it meant that a 
system is `localised in space at a given time' 
and then investigate existence as well as uniqueness 
for corresponding position operators. It turned 
out that existence and uniqueness are indeed given 
for elementary systems (fields being elements of 
irreducible representations of the Poincaré group), 
except for massless fields of higher helicity. 
A more rigorous derivation was later given by Wightman 
\cite{Wightman:1962} who also pointed out the 
connection with the representation-theoretic notion 
of `imprimitivity systems'\footnote{A good text-book 
reference explaining the notion of imprimitivity 
systems is \cite{Varadarajan:1985}.}.

It should be emphasised that the Newton--Wigner 
notion of localisation still suffers from the 
acausal spreading of localisation domains that 
is typical of fields satisfying special-relativistic 
wave equations, an observation made many times in the 
literature in one form or another; see, e.g.,  \cite{Segal.Goodman:1965,Hegerfeldt:1974,Ruijsenaars:1981}. 
This means that if a system is Newton--Wigner localised 
at a point in space at a time $t$, it is not strictly 
localised anymore in any bounded region of space at 
any time later than $t$ 
\cite{Newton.Wigner:1949,Wightman.Schweber:1955}. 
In other words, the spatial bounds of localisation do 
not develop in time within the causal future of the 
original domain. Issues of that sort, and related ones 
concerning, in particular, the relation between 
Newton--Wigner localisation and the Reeh-Schlieder 
theorem in RQFT have been discussed many times in the 
literature even up to the more recent past, 
with sometimes conflicting statements as to their 
apparent paradoxical interpretations; see, e.g., 
\cite{Fleming.Butterfield:1999} and 
\cite{Fleming:2000,Halvorson:2001}. 
For us, these issues are not in the focus of our 
interest.

Clearly, due to its historical development, 
most discussions of Newton--Wigner localisation 
put their emphasis on its relevance for RQFT. 
This sometimes seems to mask the fact that the 
problem of localisation is likewise present for 
classical systems, in particular if they are 
`relativistic' in the sense of Special Relativity. 
In fact, special-relativistic systems react with 
characteristic ambiguities if one tries to 
introduce the familiar notions of `centre of mass' 
that one uses successfully in Newtonian physics. 
A first comprehensive discussion on relativistic 
notions of `position' of various `centres' was 
given by Pryce in his 1948 paper \cite{Pryce:1948}. 
He starts with a list of no less than six different 
definitions, which Pryce labelled 
alphabetically from (a) through (f) and which for 
systems of point masses may briefly be characterised 
as follows: (a) and (c) correspond to taking 
the convex affine combination of spatial positions 
in each Lorentz frame with weights being equal to 
the rest masses and dynamical masses respectively, 
whereas (b) and (d) correspond to restricting this 
procedure to the zero-momentum frame and then 
transforming this position to other frames by 
Lorentz boosts. Possibility (e) is a combination 
of (c) and (d), determined by the condition that 
the spatial components of the ensuing position 
observable shall (Poisson) commute. This combination 
is, in fact, the Newton--Wigner position in 
its classical guise. Finally, possibility (f) is 
a variant of (b) in which the distinguished frame is 
not that of zero-momentum but that in which the 
`mass centre' as defined by (a) is at rest.

In 1965, Fleming gave a more geometric discussion 
in \cite{Fleming:1965a} that highlighted the 
group-theoretic properties (regarding the group 
of spacetime automorphisms) underlying the 
constructions and thereby clarified many of the 
sometimes controversial issues regarding 
`covariance'. Fleming focussed on three position 
observables which he called `centre of inertia', 
`centre of mass', and the Newton--Wigner position 
observable, for which he, at the very end of his 
paper and almost in passing, suggested the name 
`centre of spin'. In our paper we shall give a 
more detailed geometric justification for that 
name.

Pryce, Fleming, and other contemporary 
commentators mainly had RQFT in mind as the 
main target for their considerations, presumably 
because the study of deeply relativistic 
\emph{classical} systems was simply not considered 
relevant at that time. But that has clearly 
changed with the advent of modern relativistic 
astrophysics. For example, modern analytical studies 
of close compact binary-star systems also make use 
of various definitions of `centre of mass' in an 
attempt to separate the `overall' from the `internal' 
motion as far as possible. Note that, as is well 
known, special-relativistic many-particle systems 
will generally show dynamical couplings of internal 
and external degrees of freedom which cannot be 
eliminated altogether by more clever choices of 
external coordinates. But, in that respect, it 
turns out that modern treatments of gravitationally 
interacting two-body systems within the theoretical 
framework of Hamiltonian General Relativity show 
a clear preference for the Newton--Wigner position 
\cite{Steinhoff:2011,Schaefer.Jaranowski:2018}, 
emphasising once more its distinguished role, 
now in a purely classical context. In passing 
we mention the importance and long history 
connected with the `problem of motion' in 
General Relativity, i.e.\ the problem of how 
to associate a timelike worldline with the 
field-theoretic evolution of an extended and 
structured body, a glimpse of which may be 
obtained by the recent collection 
\cite{Puetzfeld.EtAl:2015}. A concise account 
of the various definitions of `centres' that 
have been used in the context of General 
Relativity is given in \cite{Costa.EtAl:2018}, which 
also contains most of the original references 
in its bibliography. In our opinion, all this 
provides sufficient motivation for further attempts 
to work out the characteristic properties of 
Newton--Wigner localisation in the \emph{classical} realm.

The plan of our investigation is as follows. 
After setting up our notation and conventions in 
section \ref{sec:notation}, where we also 
introduce some mathematical background, we 
prove a few results in 
section \ref{sec:NW_centre_of_spin} which are intended to 
explain in what sense the Newton--Wigner position 
is indeed a `centre of spin' and in what sense 
it is uniquely so (theorem \ref{thm:NW_SSC_centre_of_spin}). 
We continue in section \ref{sec:NW_theorem} 
with the statement and proof of a classical 
analogue of the Newton--Wigner theorem, according to 
which the Newton--Wigner position is the 
unique observable satisfying a set of axioms. 
The result is presented in theorem \ref{thm:NW} 
and in a slightly different formulation in theorem \ref{thm:NW2}. 
They say that for a classical elementary 
Poincaré-invariant system with timelike four-momentum 
(as classified by Arens \cite{Arens:1971a,Arens:1971b}), 
there is a unique observable transforming 
`as a position should' under translations, rotations, 
and time reversal, having Poisson-commuting 
components, and satisfying a regularity condition 
(being $C^1$ on all of phase space). This observable 
is the Newton--Wigner position.

\section{Notation and conventions}
\label{sec:notation}

This section is meant to list our notation and 
conventions in the general sense, by also providing 
some background material on the geometric and 
group-theoretic setting onto which the following 
two sections are based.

\subsection{Minkowski spacetime and the Poincaré group}

We use the `mostly plus' $(-{++}+)$ signature 
convention for the spacetime metric and stick, 
as indicated, to four dimensions. This is not to say 
that our analysis cannot be generalised to other 
dimensions. In fact, as will become clear as we 
proceed, many of our statements have an obvious 
generalisation to other, in particular higher 
dimensions. On the other hand, as will also become 
clear, there are a few constructions which would 
definitely look different in other dimensions, 
like, e.g., the use of the Pauli--Lubański `vector' 
in section \ref{sec:PauliLubanski}, which 
becomes an $(n-3)$-form in $n$ dimensions, 
or the classification of elementary systems.

The velocity of light will be denoted by $c$, and 
\emph{not} set equal to $1$. Affine Minkowski spacetime 
will be denoted by $M$, and the corresponding vector 
space of `difference vectors' will be denoted by $V$. 
The Minkowski metric will be denoted by 
$\eta \colon V\times V \to \mathbb R$. 
The isomorphism of $V$ with its dual space 
$V^*$ induced by $\eta$ (`index lowering') will be denoted by a 
superscript `flat' symbol $\flat$, i.e.\ for a 
vector $v \in V$ the corresponding one-form is 
$v^\flat := \eta(v, \cdot) \in V^*$. The inverse 
isomorphism (`index raising') will be denoted by a superscript 
sharp symbol $\sharp$. Note that under 
a Lorentz transformation $\Lambda$, $v\in V$ transforms 
under the defining representation, 
$(\Lambda, v) \mapsto \Lambda v$, whereas its image 
$v^\flat\in V^*$ under the $\eta$-induced isomorphism 
transforms under the inverse transposed,
$(\Lambda, v^\flat) \mapsto (\Lambda^{-1})^\top v^\flat = v^\flat \circ \Lambda^{-1}$.

We fix an orientation and a time orientation on $M$. 
The (homogeneous) Lorentz group, i.e.\ the group of 
linear isometries of $(V,\eta)$, will be denoted 
by $\mathcal L := \mathsf{O}(V,\eta)$. The Poincaré 
group, i.e.\ the group of affine isometries of 
$(M,\eta)$, will be denoted by $\mathcal P$. 
The proper orthochronous Lorentz and Poincaré 
groups (i.e.\ the connected components of the identity) 
will be denoted by $\mathcal L_+^\uparrow$ and 
$\mathcal P_+^\uparrow$, respectively\footnote
	{Note that speaking of just orthochronous or proper 
	Lorentz / Poincaré transformations does not make 
	invariant sense without specifying a time direction.}.

We employ standard index notation for Minkowski 
spacetime, using lowercase Greek letters for 
spacetime indices. When working with respect to bases, 
we will, unless otherwise stated, assume them to be 
positively oriented and orthonormal, and we will 
use 0 for the timelike and lowercase Latin letters 
for spatial indices. We will adhere to standard practice 
in physics where lowering and raising of indices are 
done while keeping the same kernel symbol; i.e.\ for a 
vector $v \in V$ with components $v^\mu$, the components 
of the corresponding one-form $v^\flat \in V^*$ will be 
denoted simply by $v_\mu$. For the sake of notational 
clarity, we will sometimes denote the Minkowski 
inner product of two vectors $u, v \in V$ simply by
\begin{equation}
	u \cdot v := \eta(u,v) = u_\mu v^\mu.
\end{equation}

We fix, once and for all, a reference point / 
\emph{origin} $o \in M$ in (affine) Minkowski 
spacetime, allowing us to identify $M$ with its 
corresponding vector space $V$ (identifying the 
reference point $o \in M$ with the zero vector 
$0 \in V$, i.e.\ via $M \ni x \mapsto (x-o) \in V$), 
which we will do most of the time. Using the 
reference point $o \in M$, the Poincaré group splits 
as a semidirect product
\begin{equation}
	\mathcal P = \mathcal L \ltimes V
\end{equation}
where the Lorentz group factor in this decomposition 
arises as the stabiliser of the reference point~-- 
i.e.\ a Poincaré transformation is considered a 
homogeneous Lorentz transformation if and only if 
it leaves $o$ invariant. Thus, a homogeneous 
Lorentz transformation $\Lambda \in \mathcal L$ 
acts on a point $x \in M \equiv V$ as 
$(\Lambda x)^\mu = \Lambda^\mu_{\hphantom{\mu}\nu} x^\nu$, 
and a Poincaré transformation $(\Lambda,a) \in \mathcal P$ 
acts as $((\Lambda,a) \cdot x)^\mu = \Lambda^\mu_{\hphantom{\mu}\nu} x^\nu + a^\mu$.

We will sometimes make use of the set of spacelike 
hyperplanes in (affine) Minkowski spacetime $M$, 
which we will denote by
\begin{equation}
	\label{eq:def_SpHp}
	\mathsf{SpHP} := \{\Sigma \subset M : \Sigma \; \text{spacelike hyperplane}\}.
\end{equation}
Since the image of a spacelike hyperplane under 
a Poincaré transformation is again a spacelike 
hyperplane, there is a natural action of the Poincaré 
group on $\mathsf{SpHP}$, which we will denote by 
$((\Lambda, a), \Sigma) \mapsto (\Lambda, a) \cdot \Sigma$ 
and spell out in more detail in equation 
\eqref{eq:action_on_SpHP} below.

\subsection{The Poincaré algebra}

When considering the Lie algebra $\mathfrak p$ 
of the Poincaré group (or symplectic representations 
thereof), we will denote the generators of translations 
by $P_\mu$ such that $a^\mu P_\mu$ is the 
`infinitesimal transformation' corresponding 
to the translation by $a\in V$, and the generators 
of homogeneous Lorentz transformations (with respect 
to the chosen origin $o$) by $J_{\mu\nu}$, 
such that $- \frac{1}{2} \omega^{\mu\nu} J_{\mu\nu}$ 
is the `infinitesimal transformation' corresponding 
to the Lorentz transformation $\exp(\omega)\in\mathcal
L_+^\uparrow \subset \mathsf{GL}(V)$ for 
$\omega \in \mathfrak l = \mathrm{Lie}(\mathcal L) \subset \mathrm{End}(V)$.

Since we are using the $(-{++}+)$ signature 
convention, the minus sign in the expression 
$- \frac{1}{2} \omega^{\mu\nu} J_{\mu\nu}$ is 
necessary in order that $J_{ab}$ generate rotations 
in the $e_a$--$e_b$ plane from $e_a$ towards $e_b$, 
which is the convention we want to adopt. A detailed 
discussion of these issues regarding sign conventions 
for the generators of special orthogonal groups 
can be found in appendix \ref{app:sign_convention_so}. 
Moreover, if $u \in V$ is a future-directed unit 
timelike vector, then $c P_\mu u^\mu$ (i.e.\ $c P_0$ 
in the Lorentz frame defined by $u=e_0$), which is 
\emph{minus} the energy in the frame defined by $u$, 
is the generator of active time translations in the 
direction of $u$. Therefore, with our conventions, 
for the case of causal four-momentum $P \in V$ the energy 
(with respect to future-directed time directions) 
is positive if and only if $P$ is future-directed.

With our conventions, the commutation relations 
for the Poincaré generators are as follows:
\begin{subequations} 
\label{eq:Poinc_gen}
\begin{align}
\label{eq:Poinc_gen-a}
	[P_\mu, P_\nu] &= 0 \\
\label{eq:Poinc_gen-b}
	[J_{\mu\nu}, P_\rho] &= \eta_{\mu\rho} P_\nu - \eta_{\nu\rho} P_\mu \\
\label{eq:Poinc_gen-c}
	[J_{\mu\nu}, J_{\rho\sigma}] &= \eta_{\mu\rho} J_{\nu\sigma} + \text{(antisymm.)} \nonumber\\
		&= \Big(\eta_{\mu\rho} J_{\nu\sigma} - (\mu \leftrightarrow \nu)\Big) - \Big(\rho \leftrightarrow \sigma\Big)
\end{align}
\end{subequations}
As indicated, the abbreviation `antisymm.', which 
we shall also use in the sequel of this paper,  
stands for the additional three terms that one obtains 
by first antisymmetrising (without a factor of $1/2$) 
in the first pair of indices on the left hand 
side, here $(\mu\nu)$, and then the ensuing 
combination once more in the second set of 
indices, here $(\rho\sigma)$, again without a 
factor $1/2$.

For later reference we already point out here that 
this Lie algebra has several convenient features, one 
of which being that it is \emph{perfect}. This 
means that it equals its own derived algebra 
or, in other words, that each of its element is 
expressible as a linear combination of Lie brackets. 
This is easy to see directly from \eqref{eq:Poinc_gen}. 
Indeed, contraction of  \eqref{eq:Poinc_gen-b} and 
\eqref{eq:Poinc_gen-c} with $\eta^{\mu\rho}$ gives 
$(\dim V-1)P_\nu$ in the first and 
$(\dim V-2)J_{\nu\sigma}$ in the second case, showing that 
each basis element $P_\nu$ and  $J_{\nu\sigma}$ is 
a linear combination of Lie brackets if $\dim V>2$. 
Being perfect implies that its first cohomology is 
trivial. Moreover, the second cohomology is also 
trivial. Being perfect and of trivial second 
cohomology will later allow us to conclude 
that symplectic actions are necessarily Poisson 
actions.

\subsection{Symplectic geometry}

We employ the following sign conventions for symplectic
geometry (as used by Abraham and Marsden in 
\cite{Abraham.Marsden:1978}, but different to those 
of Arnold in \cite{Arnold:1989}). Let $(\Gamma, \omega)$ 
be a symplectic manifold. For a smooth function 
$f \in C^\infty(\Gamma)$, we define the Hamiltonian 
vector field $X_f\in ST(\Gamma)$ ($ST$ denoting 
sections in the tangent bundle) corresponding to 
$f$ by
\begin{equation} \label{eq:def_Ham_VF}
	\iota_{X_f} \omega := \omega(X_f, \cdot) = \D f,
\end{equation}
where $\iota$ denotes the interior product 
between vector fields and differential forms.
The Poisson bracket of two smooth functions 
$f,g \in C^\infty(\Gamma)$ is then defined as
\begin{equation}
	\{f, g\} := \omega(X_f, X_g) = \D f(X_g) = \iota_{X_g} \D f.
\end{equation}
These conventions give the usual coordinate forms 
of the Hamiltonian flow equations and the Poisson 
bracket if the symplectic form $\omega$ takes the 
coordinate form (sign-opposite to that in \cite{Arnold:1989})
\begin{equation}
	\omega = \D q^a \wedge \D p_a\,.
\end{equation}
It is important to note that $C^\infty(\Gamma)$ 
as well as $ST(\Gamma)$ are (infinite dimensional) 
Lie algebras with respect to the Poisson bracket 
and the commutator respectively, and that, with 
respect to these Lie structures, the map 
$C^\infty(\Gamma) \to ST(\Gamma), f \mapsto X_f$ 
is a Lie \emph{anti}-homomorphism, that is,  
\begin{equation} \label{eq:HVF_anti_hom}
	X_{\{f,g\}} = -\,[X_f,X_g].	
\end{equation}
The proof is simple once one recalls from 
\eqref{eq:def_Ham_VF} that the Lie derivative of 
$\omega$ with respect to any Hamiltonian vector 
field vanishes: 
$L_{X_f} \omega = \D(\iota_{X_f} \omega) + \iota_{X_f} \D\omega = \D^2 f = 0$.
Therefore, 
$\D\{f,g\} = \D(\iota_{X_g} \D f) = L_{X_g} \D f =
L_{X_g}(\iota_{X_f} \omega) = - \iota_{[X_f,X_g]} \omega$.

By saying that a one-parameter group 
$\phi_s \colon \Gamma \to \Gamma$ of symplectomorphisms 
is generated by a function $g \in C^\infty(\Gamma)$, 
we mean that $\phi_s$ is the flow of the Hamiltonian 
vector field to $g$, i.e.\ that
\begin{equation}
	\frac{\D}{\D s} \phi_s(\gamma) = X_g(\phi_s(\gamma))
\end{equation}
for $\gamma \in \Gamma$, or equivalently
\begin{align} \label{eq:sympl_gen_PB}
	\frac{\D}{\D s} (f \circ \phi_s) &= \Big(\D f (X_g)\Big) \circ \phi_s \nonumber\\
	&= \{f,g\} \circ \phi_s
\end{align}
for $f \in C^\infty(\Gamma)$. Here both sides of
\eqref{eq:sympl_gen_PB} are to be understood as 
evaluated pointwise.

\subsection{Poincaré-invariant Hamiltonian systems and their momentum maps}

A classical Poincaré-invariant system will be described 
by a phase space $(\Gamma, \omega)$ -- i.e.\ a symplectic 
manifold -- with a symplectic action
\begin{equation}
	\Phi \colon \mathcal P \times \Gamma \to \Gamma, \; ((\Lambda, a), \gamma) \mapsto \Phi_{(\Lambda, a)} (\gamma)
\end{equation}
of the Poincaré group (in fact, for most of our purposes 
an action of $\mathcal P_+^\uparrow$ is enough). We will 
take $\Phi$ to be a `left' action, i.e.\ to satisfy\footnote
	{We refer to \cite{Giulini:2015} for a detailed discussion 
	of left versus right actions and the corresponding sign 
	conventions that will also play an important role in 
	the sequel of this paper.}
\begin{equation}
	\Phi_{(\Lambda_1, a_1)} \circ \Phi_{(\Lambda_2, a_2)} = \Phi_{(\Lambda_1 \Lambda_2, a_1 + \Lambda_1 a_2)} \; .
\end{equation}
We will denote such systems as 
$(\Gamma, \omega, \Phi)$.

The left action $\Phi$ of $\mathcal{P}$ on $\Gamma$
induces vector fields $V_\xi$ on $\Gamma$ (the 
so-called `fundamental vector fields'), one for 
each $\xi$ in the Lie algebra $\mathfrak{p}$ of 
$\mathcal{P}$. They are given by
\begin{equation}
	V_\xi(\gamma) := \left. \frac{\D}{\D s} \Phi_{\exp(s\xi)}(\gamma) \right\vert_{s=0} \; ,	
\end{equation}  
so that the map $\mathfrak{p} \to ST(\Gamma), 
\xi\mapsto V_\xi$, given by the differential 
of $\Phi$ with respect to its first argument and 
evaluated at the group identity, is clearly linear. 
In fact, it is straightforward to show that it 
is an anti-homomorphism from the Lie algebra 
$\mathfrak{p}$ into the Lie algebra $ST(M)$\footnote
	{Had we chosen $\Phi$ to be a right action, we would 
	have obtained a proper Lie homomorphism; compare 
	\cite[appendix B]{Giulini:2015}.}, i.e.
\begin{equation} \label{eq:fundVF_anti_hom}
	\bigl[V_{\xi_1}, V_{\xi_2}\bigr] = - V_{[\xi_1, \xi_2]}.
\end{equation}
Moreover, a similar calculation shows 
\cite[appendix B]{Giulini:2015}
\begin{equation} \label{eq:fundVF_equivariance}
	(\DD \Phi_{(\Lambda,a)}) \circ V_\xi = V_{\mathrm{Ad}_{(\Lambda,a)}(\xi)} \circ \Phi_{(\Lambda,a)} \; ,
\end{equation}
where $\DD \Phi_{(\Lambda,a)} \colon T\Gamma \to T\Gamma$ 
denotes the differential of 
$\Phi_{(\Lambda,a)} \colon \Gamma \to \Gamma$.

As $\mathcal{P}$ acts by symplectomorphisms, 
we clearly have 
\begin{equation}
	L_{V_\xi} \omega = 0 \quad \text{for all} \; \xi \in \mathfrak{p}.
\end{equation}
As $\omega$ is closed, the latter equation implies 
that $\iota_{V_\xi} \omega$ is likewise closed. 
Hence, by Poincaré's lemma, locally (i.e.\ in a 
neighbourhood of each point) there exists a local 
function $f_\xi$ such that $\D f_\xi = \iota_{V_\xi} \omega$. 
This function is unique up to the addition of a 
$\xi$-dependent constant. Again by Poincaré's lemma 
we could argue that $f_\xi$ existed globally if 
$\Gamma$ were simply connected. But, fortunately, 
we do not need that extra assumption.

In fact, since we are dealing with a special 
group, the function $f_\xi$ always exists 
\emph{globally}, irrespective of $\Gamma$'s 
topology, so that each $V_\xi$ is a globally 
defined Hamiltonian vector field (i.e.\ each 
one-parameter group 
$\Phi_{\exp(s\xi)}\colon \Gamma \to \Gamma$ of
symplectomorphisms is generated, in the sense 
of \eqref{eq:sympl_gen_PB}, by the corresponding
function $f_\xi$). Moreover,  
the constants up to which the collection of 
$f_\xi$ is defined can be chosen in such a 
way that the map $\xi\mapsto f_\xi$ from the 
Lie algebra $\mathfrak{p}$ to the Lie 
algebra $C^\infty(\Gamma)$ (the Lie product 
of the latter being the Poisson bracket) is a 
Lie homomorphism, i.e.\ 
\begin{equation} \label{eq:Poinc_gen_hom}
	\left\{f_{\xi_1},f_{\xi_2}\right\} = f_{[\xi_1,\xi_2]}.
\end{equation}
This clearly fixes the constants uniquely.  
Note that, according to \eqref{eq:fundVF_anti_hom} 
and \eqref{eq:HVF_anti_hom}, both maps, 
$\xi\mapsto V_\xi$ and $V_\xi\mapsto f_\xi$, are 
Lie \emph{anti}-homomorphisms. Hence their combination 
$\xi\mapsto f_\xi$ is a proper Lie homomorphism 
(no minus sign on the right-hand side of 
\eqref{eq:Poinc_gen_hom}).

A symplectic action of a group whose generating 
vector fields are globally Hamiltonian and satisfy 
\eqref{eq:Poinc_gen_hom} is called a 
\emph{Poisson action}. The statement made here 
is that if $\dim V>2$, any symplectic action of 
the Poincaré of group is always a Poisson action.
 This is a non-trivial statement 
depending crucially on properties of the groups's 
Lie algebra. For example, it would fail to hold for 
the Galilei group (homogeneous as well as 
inhomogeneous) which, despite being just a 
contraction of the Poincaré group, behaves 
quite differently in that matter and, consequently, 
also as regards the problem of localisation 
\cite{Inoenue.Wigner:1952,Wightman:1962}.

The underlying reason for why $f_\xi$ exists 
globally is that $\mathfrak{p}$ 
is perfect, as already shown above. Indeed, the 
proof is quite simple: Since $\xi = [\xi_1,\xi_2]$ 
(or sums of such commutators) we have 
$V_\xi = -[V_{\xi_1},V_{\xi_2}]$ and hence 
$\D f_\xi = -\iota_{[V_{\xi_1},V_{\xi_2}]} \omega = 
- L_{V_{\xi_1}} (\iota_{V_{\xi_2}} \omega) = 
\D(\omega(V_{\xi_1},V_{\xi_2}))$, 
so that 
$f_\xi = \omega(V_{\xi_1},V_{\xi_2})+\text{const.}$ 
which is globally defined. The other statement 
concerning the choice of constants that guarantee 
\eqref{eq:Poinc_gen_hom} is an immediate consequence 
of the triviality of the second cohomology of 
$\mathfrak{p}$, the proof of which may, e.g., be 
looked up in \cite[\S\,3.3]{Woodhouse:1980}.

Having established global existence and uniqueness
of the generators $f_\xi$ satisfying
$\omega(V_\xi, \cdot) = \D f_\xi$, we can now 
deduce the transformation property of $f_\xi$
under the action of $\mathcal{P}$. Taking the 
pullback of the equation $\omega(V_\xi, \cdot) = \D f_\xi$ with 
$\Phi_{(\Lambda,a)^{-1}}$ and using the invariance of 
$\omega$ as well as \eqref{eq:fundVF_equivariance}, we 
immediately deduce
\begin{equation} \label{eq:Poinc_gen_Ad_equivariance}
	\Phi_{(\Lambda,a)^{-1}}^* f_\xi := f_\xi \circ \Phi_{(\Lambda,a)^{-1}} = f_{\mathrm{Ad}_{(\Lambda,a)}(\xi)} \; ,
\end{equation}
which may also be read as the invariance of the 
real-valued function $f\colon \mathfrak{p} \times \Gamma \to \mathbb R$, 
$(\xi,\gamma) \mapsto f_\xi(\gamma)$, under the 
combined left action of $\mathcal{P}$ 
on $\mathfrak{p}\times\Gamma$ given by 
$\mathrm{Ad}\times\Phi$. Alternatively, since 
$\xi\mapsto f_\xi$ is linear, we may regard 
$f$ as $\mathfrak{p}^*$-valued function on 
$\Gamma$, where $\mathfrak{p}^*$ denotes the 
vector space dual to $\mathfrak{p}$. 
This map is called the \emph{momentum map}\footnote
	{See \cite[chap.\,4.2]{Abraham.Marsden:1978} 
	for a general discussion on the notion of `momentum map'
	and also \cite{Giulini:2015} for an account of its 
	use and properties restricted to the case of 
	Poincaré-invariant systems.}
for the given system $(\Gamma,\omega,\Phi)$, 
which according to \eqref{eq:Poinc_gen_Ad_equivariance} is then
$\mathrm{Ad}^*$-equivariant:
\begin{equation} \label{eq:momentum_map_Ad-star_equivariance}
	f \circ \Phi_{(\Lambda,a)} = \mathrm{Ad}^*_{(\Lambda,a)} \circ f
	\iff
	\mathrm{Ad}^*_{(\Lambda,a)} \circ f \circ \Phi_{(\Lambda,a)^{-1}} = f
\end{equation}
The second expression is again meant to 
stress that the condition of equivariance is 
equivalent to the invariance of the function $f$ 
under the combined left actions in its domain 
and target spaces (invariance of the graph). 
Note that $\mathrm{Ad}^*$ denotes the co-adjoint representation of  $\mathcal{P}$ on $\mathfrak{p}^*$, 
given by $\mathrm{Ad}^*_{(\Lambda,a)} := (\mathrm{Ad}_{(\Lambda,a)^{-1}})^\top$ with superscript $\top$
denoting the transposed map.

Points in $\Gamma$ faithfully represent the state 
of the physical system whereas observables 
correspond to functions on $\Gamma$. In order to 
implement time evolution we shall employ a 
`classical Heisenberg picture', in which the 
phase space point remains the same at all 
times, whereas the evolution will correspond to 
the changes of observables according to their 
association to different spacelike 
hyperplanes in spacetime. Although this is 
different from the (`Schrödinger picture') 
approach usually taken in classical mechanics 
(where the state of the system is given by a 
phase space point changing in `time', which is 
an external parameter), this point of view is 
clearly better adapted to the 
Poincaré-relativistic framework, in which 
there simply is no absolute notion of time.

Choosing a set 
of ten basis vectors $(P_\mu,J_{\mu\nu})$ for 
$\mathfrak{p}$ obeying \eqref{eq:Poinc_gen} (compare 
appendix \ref{app:sign_convention_so}), we can contract the 
$\mathfrak{p}^*$-valued momentum map with each of 
these basis vectors in order to obtain the 
corresponding ten real-valued component functions 
of the momentum map. By some abuse of notation 
we shall call these component functions by the same 
letters  $(P_\mu,J_{\mu\nu})$ as the Lie algebra 
elements themselves. Equation 
\eqref{eq:Poinc_gen_hom} now says that the map 
that sends the Lie algebra elements 
$P_\mu$ and $J_{\mu\nu}$ in $\mathfrak{p}$ 
to the corresponding component functions 
of the momentum map is a Lie homomorphism from 
$\mathfrak{p}$ to the Lie algebra 
$C^\infty(\Gamma,\mathbb{R})$ (the latter with 
Poisson bracket as Lie multiplication):
\begin{subequations}
	\begin{align}
	\{P_\mu, P_\nu\} &= 0 \\
	\{J_{\mu\nu}, P_\rho\} &= \eta_{\mu\rho} P_\nu - \eta_{\nu\rho} P_\mu \\
	\{J_{\mu\nu}, J_{\rho\sigma}\} &= \eta_{\mu\rho} J_{\nu\sigma} + \text{(antisymm.)}
	\end{align}
\end{subequations}
The $\mathrm{Ad}^*$-equivariance of the momentum map can now be 
written down in component form if we first set 
$\xi = P_\mu$ and then $\xi = J_{\mu\nu}$. 
Indeed, considering
\eqref{eq:Poinc_gen_Ad_equivariance} and 
recalling our abuse of notation in 
denoting the real-valued phase space functions 
$f_{P_\mu}$ and $f_{J_{\mu\nu}}$
again with the letters $P_\mu$ and  
$J_{\mu\nu}$, we can immediately read from 
equation \eqref{eq:ad_rep} of appendix 
\ref{app:adj_rep}, in which we need to 
replace $e_a$ with $P_\mu$ and $B_{ab}$ 
with $-J_{\mu\nu}$ according to 
\eqref{eq:sign_convention_so} of 
appendix \ref{app:sign_convention_so}, that
\begin{subequations} \label{eq:coadjoint_rep_components}
\begin{alignat}{2}
\label{eq:coadjoint_rep_components_a}
	& P_\mu \circ \Phi_{(\Lambda, a)}
	&&= (\Lambda^{-1})^\nu_{\phantom{\nu}\mu} \, P_\nu \; , \\
\label{eq:coadjoint_rep_components_b}
	& J_{\mu\nu} \circ \Phi_{(\Lambda, a)}
	&&= (\Lambda^{-1})^\rho_{\phantom{\rho}\mu} (\Lambda^{-1})^\sigma_{\phantom{\sigma}\nu} \,
	J_{\rho\sigma} 
	+ a_\mu (\Lambda^{-1})^\rho_{\phantom{\rho}\nu} \, P_\rho
	- a_\nu (\Lambda^{-1})^\rho_{\phantom{\rho}\mu} \, P_\rho \; .
\end{alignat}
\end{subequations}
Note that the left-hand sides of \eqref{eq:coadjoint_rep_components} are 
precisely what we need; that is, we need 
the composition with $\Phi_{(\Lambda, a)}$ rather than 
$\Phi_{(\Lambda, a)^{-1}}$ to evaluate the momenta 
$P_\mu$ and $J_{\mu\nu}$ on the actively 
Poincaré-displaced phase space points. Note also 
that if we had put the 
indices upstairs and had used, e.g., 
$P^\mu = \eta^{\mu\nu} P_{\nu}$ rather than $P_\mu$ 
then the right-hand side of \eqref{eq:coadjoint_rep_components_a} would read 
$\Lambda^\mu_{\phantom{\mu}\nu} \, P^\nu$, and 
correspondingly in 
\eqref{eq:coadjoint_rep_components_b}. 
Finally recall that the last term on the right-hand 
side of \eqref{eq:coadjoint_rep_components_b} just 
reflects the familiar transformation of angular 
momentum (the momentum associated to spatial 
rotations) under spatial translations, which is 
typical for the \emph{co}-adjoint representation, 
which here gets extended to the momentum 
associated to boost transformations\footnote
	{One easily checks that the signs are 
	right: translating a system whose momentum points 
	in $y$-direction by a positive amount into the 
	$x$-direction should enhance the angular momentum 
	in $z$-direction. This is just what 
	\eqref{eq:coadjoint_rep_components_b} implies.}.

\subsection{The Pauli--Lubański vector}
\label{sec:PauliLubanski}
Given a classical Poincaré-invariant system, the 
Pauli--Lubański vector $W$ is the $V$-valued phase 
space function defined in components by
\begin{equation}
	W_\mu = - \frac{1}{2} \varepsilon_{\mu\nu\rho\sigma} P^\nu J^{\rho\sigma}
\end{equation}
where $\varepsilon$ denotes the volume form of 
Minkowski space (whose components in a positively 
oriented orthonormal basis are just given by the 
usual totally antisymmetric symbol, with 
$\varepsilon_{0123} = +1$). The sign convention 
in this definition can be understood as follows. 
We imagine a situation in which $P$ is timelike and 
future-directed (positive energy, see above), 
and consider the spatial components of $W$ with 
respect to an orthonormal basis $\{e_0, \dots, e_3\}$ 
of $V$ with $(e_0)^\mu = P^\mu / \sqrt{-P_\nu P^\nu}$ 
(`momentum rest frame'). For those, we obtain
\begin{equation}
	\frac{W_a}{\sqrt{-P_\mu P^\mu}} = - \frac{1}{2} \varepsilon_{a0\rho\sigma} J^{\rho\sigma} = \frac{1}{2} {^{(3)}\varepsilon}_{abc} J^{bc}
\end{equation}
where the ${^{(3)}\varepsilon}_{abc}$ is the 
three-dimensional antisymmetric symbol / the 
components of the spatial volume form. Thus, since 
$J^{bc} = J_{bc}$ generates rotations from $e_b$ 
towards $e_c$, we see that $W_a/\sqrt{-P_\mu P^\mu}$ 
generates rotations `along the $e_a$ axis' in the usual, 
three-dimensional sense. Thus, $W/\sqrt{-P_\mu P^\mu}$ 
can be interpreted as the `spatial spin vector' in the 
momentum rest frame, which is the usual interpretation 
of the Pauli--Lubański vector.

Rewriting the definition of $W$ as
\begin{equation}
	W_\mu = - \frac{1}{2} \varepsilon_{\mu\nu\rho\sigma} P^\nu J^{\rho\sigma} = \frac{1}{2} \varepsilon_{\nu\rho\sigma\mu} P^\nu J^{\rho\sigma} = \frac{1}{3!} \varepsilon_{\nu\rho\sigma\mu} (P^\flat \wedge J)^{\nu\rho\sigma},
\end{equation}
we see that in the language of exterior algebra
\begin{equation}
	W = (*(P^\flat\wedge J))^\sharp
\end{equation}
where $*$ is the Hodge star operator. Here we 
use the standard sign conventions for the Hodge 
operator, i.e.\ the definition 
$\alpha \wedge *\beta = \eta(\alpha,\beta) \, \varepsilon$; 
see for example \cite{Straumann:2013} or \cite[appendix A]{Giulini:2015}.

\section{The Newton--Wigner position as a `centre of spin'}
\label{sec:NW_centre_of_spin}

In this section we will explain our understanding 
and present our geometric clarification of 
Fleming's statement in \cite{Fleming:1965a} that 
the Newton--Wigner position may be understood as 
a `centre of spin'. To this end, we introduce 
Fleming's geometric framework for special-relativistic 
position observables, and then discuss the definition 
of position observables by spin supplementary conditions 
(SSCs). Finally, we introduce the notion of a position 
observable being a `centre of spin', and prove that 
the Newton--Wigner position is the only continuous 
position observable defined by an SSC that represents 
a centre of spin in that sense.

\subsection{Position observables on spacelike hyperplanes}
\label{sec:pos_obs}

We start by describing the general framework 
developed by Fleming in \cite{Fleming:1965a} 
and also \cite{Fleming:1966} for the description 
of special-relativistic position observables, 
translated to our case of classical systems from 
Fleming's quantum language. Consider a classical 
Poincaré-invariant system $(\Gamma, \omega, \Phi)$. 
By a \emph{position observable} $\chi$ for this 
system we understand a `procedure' which, given 
any spacelike hyperplane $\Sigma \in \mathsf{SpHP}$ 
in (affine) Minkowski spacetime, allows us to 
`localise' the system on $\Sigma$. More precisely, 
this means that for any $\Sigma \in \mathsf{SpHP}$, 
we have an $M$-valued phase space function
\begin{equation}
	\chi(\Sigma) \colon \Gamma \to M
\end{equation}
with image contained in $\Sigma$, whose value 
$\chi(\Sigma)(\gamma)$ for $\gamma \in \Gamma$ is 
to be interpreted as the `$\chi$-position' of our 
system in state $\gamma$ on the hyperplane $\Sigma$.

Any spacelike hyperplane $\Sigma \in \mathsf{SpHP}$ 
is uniquely characterised by its (timelike) 
future-directed unit normal $u \in V$ and its distance 
$\tau\in\mathbb R$ to the origin $o \in M$, measured 
along the straight line through $o$ in direction $u$.  
In terms of these, it has the form 
\begin{equation} \label{eq:hyperplane_params}
	\Sigma = \{x\in M : u_\mu x^\mu = -\tau\},
\end{equation}
where we identified $M$ with $V$. 
From now on, whenever convenient, we will identify 
$\Sigma$ with the tuple $(u, \tau)$. The condition 
that the image of $\chi(\Sigma)$ be contained 
in $\Sigma$ then takes the form
\begin{equation} \label{eq:pos_obs_image}
	u_\mu \chi^\mu(u, \tau)(\gamma) = - \tau.
\end{equation}
We can now also spell out explicitly the left action 
of $\mathcal{P}$ on $\mathsf{SpHP}$ that is induced 
from the left action of $\mathcal{P}$ on $M$
(as already mentioned below equation \eqref{eq:def_SpHp}):
\begin{equation} \label{eq:action_on_SpHP}	
	(\Lambda,a) \cdot (u,\tau) = (\Lambda u, \tau - \Lambda u \cdot a)
\end{equation}
One easily checks that this indeed defines a left action, i.e.\ 
$(\Lambda_1,a_1)\cdot [(\Lambda_2,a_2)\cdot(u,\tau)] = 
(\Lambda_1\Lambda_2,a_1+\Lambda_1a_2)\cdot(u,\tau)$.

Fixing $u$ and varying $\tau$ in 
\eqref{eq:hyperplane_params}, we obtain the spacelike 
hyperplanes corresponding to different `instants of 
time' $\tau$ in the Lorentz frame corresponding to 
$u$. Thus, for a fixed state $\gamma\in\Gamma$ and fixed 
frame $u$, the set
\begin{equation}
	\{\chi(u, \tau) (\gamma) : \tau \in \mathbb R\} \subset M
\end{equation}
gives the `worldline' of the $\chi$-position of 
the system. Following Fleming \cite{Fleming:1965a}, who 
says that this is a requirement `easily agreed upon', 
we require that this worldline should be parallel to 
the four-momentum\footnote
	{This assumption is natural for closed systems as we 
	consider here. For non-closed systems, i.e.\ systems 
	without local energy--momentum conservation, the 
	four-velocity is in general not parallel to the 
	four-momentum; see, e.g., the discussion at the beginning 
	of section 2.6 in \cite{Giulini:2018}.},
i.e.\ 
$\frac{\partial \chi(u, \tau)}{\partial \tau} \propto P$. 
Together with \eqref{eq:pos_obs_image}, this implies 
condition \eqref{eq:pos_obs_time_der} in the 
definition below, which is meant to sum up all the 
preceding considerations.
\begin{defn} \label{defn:pos_obs}
	A \emph{position observable} for a classical 
	Poincaré-invariant system $(\Gamma, \omega, \Phi)$ 
	with causal four-momentum is a map
	\begin{equation} \label{eq:pos_obs}
		\chi \colon \mathsf{SpHP} \times \Gamma \to M, \; (\Sigma, \gamma) \mapsto \chi(\Sigma)(\gamma)
	\end{equation}
	satisfying
	\begin{equation}
		\chi(\Sigma)(\gamma) \in \Sigma
	\end{equation}
	for all $\Sigma \in \mathsf{SpHP}$ and all 
	$\gamma \in \Gamma$ (or, equivalently, 
	\eqref{eq:pos_obs_image}), as well as
	\begin{equation} \label{eq:pos_obs_time_der}
		\frac{\partial \chi_\mu(u, \tau)}{\partial \tau} = \frac{1}{(- u \cdot P)} P_\mu \; .
	\end{equation}
	For fixed $\Sigma\in\mathsf{SpHP}$, we will often 
	view $\chi(\Sigma) \colon \Gamma \to M$ as a 
	phase space function in its own right.
\end{defn}
Note that \eqref{eq:pos_obs_time_der} and 
\eqref{eq:pos_obs_image} imply that the 
four-momentum must be causal for such a position 
observable to exist.

In addition to the demands of the positions 
$\chi(\Sigma)$ being located on $\Sigma$ and 
of `worldlines' in direction of the four-momentum, 
Fleming also introduces the following covariance 
requirement (which we, different to Fleming, do not 
include in the definition of a position observable):
\begin{defn} \label{defn:pos_obs_cov}
	A position observable for a classical 
	Poincaré-invariant system $(\Gamma, \omega, \Phi)$ 
	is said to be \emph{covariant} if and only if
	\begin{equation} \label{eq:pos_obs_cov_1}
		\chi \Big((\Lambda, a) \cdot \Sigma\Big) \Big(\Phi_{(\Lambda, a)}(\gamma)\Big) = (\Lambda, a) \cdot \Big( \chi(\Sigma)(\gamma) \Big)
	\end{equation}
	for all $\Sigma\in\mathsf{SpHP}$, $\gamma \in \Gamma$ 
	and $(\Lambda, a) \in \mathcal P$. This can be 
	read concisely as saying that the map 
	\eqref{eq:pos_obs} is invariant under 
	the natural left action induced from those on 
	the domain and target spaces (invariance of 
	$\chi$'s graph):
	\begin{equation} \label{eq:pos_obs_cov_2}
		\chi = (\Lambda,a) \circ \chi \circ \left( (\Lambda,a)^{-1} \times \Phi_{(\Lambda,a)^{-1}} \right).
	\end{equation}
\end{defn}

This is indeed a sensible notion of covariance: 
it demands that, for any Poincaré transformation 
$(\Lambda, a)$, the $\chi$-position of the 
transformed system $\Phi_{(\Lambda, a)}(\gamma)$ 
on the transformed hyperplane 
$(\Lambda, a) \cdot \Sigma$ be the transform 
of the `original position' $\chi(\Sigma)(\gamma)$. 
In terms of components, \eqref{eq:pos_obs_cov_1} 
assumes the form
\begin{equation}
	\chi^\mu (\Lambda u, \tau - \Lambda u\cdot a) \circ \Phi_{(\Lambda, a)} = \Lambda^\mu_{\hphantom{\mu}\nu} \chi^\nu(u, \tau) + a^\mu\,,
\end{equation}
taking into account \eqref{eq:action_on_SpHP}.

\subsection{Spin supplementary conditions}

The most important and widely used procedure to 
define special-relativistic position observables 
is by so-called \emph{spin supplementary conditions}. 
Suppose we are given a causal, future-directed 
vector $P \in V$ and an antisymmetric 2-tensor 
$J \in \bigwedge^2 V^*$, describing the four-momentum 
and the angular momentum (with respect to the origin 
$o \in M$) of some physical system. For any 
future-directed timelike vector $f \in V$, 
we then consider the equation
\begin{equation} \label{eq:SSC}
	0 = S_{\mu\nu} f^\nu
\end{equation}
with $S_{\mu\nu} := J_{\mu\nu} - x_\mu P_\nu + x_\nu P_\mu$, 
which we view as an equation for $x \in M$. Since 
$S$ is the angular momentum tensor with respect 
to the reference point $x$ (instead of the origin 
$o$ as for $J$), or the \emph{spin tensor} with 
respect to $x$, \eqref{eq:SSC} is called the 
\emph{spin supplementary condition} (SSC) with 
respect to $f$. As is well-known (and easily 
verified), the set of its solutions $x$ is a 
line in $M$ with tangent $P$, namely
\begin{equation} \label{eq:SSC_worldline}
	\{x \in M : 0 = S_{\mu\nu} f^\nu\} = \left\{x \in M : x_\mu = \frac{J_{\mu\rho} f^\rho}{f\cdot P} + \lambda P_\mu \; \text{with} \; \lambda \in \mathbb R\right\}.
\end{equation}
This line can be given the interpretation of the 
`centre of energy' worldline of our system with 
respect to the Lorentz frame defined by $f$. 
See \cite{Costa.EtAl:2018} and references therein 
for further discussion on the interpretation and 
impact of various SSCs as regards equations of motion 
in General Relativity.

The idea is now to explicitly combine the 
SSC-based approach with Fleming's geometric ideas, 
thereby introducing the two independent parameters 
$f$ from \eqref{eq:SSC_worldline} and $u$ from 
\eqref{eq:hyperplane_params}. We define a position
observable in the sense of definition \ref{defn:pos_obs} 
in the following way: given a classical 
Poincaré-invariant system $(\Gamma, \omega, \Phi)$ 
with causal four-momentum and a state 
$\gamma \in \Gamma$, we consider the SSC worldline
defined by \eqref{eq:SSC} where we now take
$P_\mu(\gamma)$ for the four-momentum and 
$J_{\mu\nu}(\gamma)$ for the angular momentum tensor. 
We then simply define $\chi(\Sigma)(\gamma)$ to be 
the intersection of this worldline with the 
hyperplane $\Sigma = (u,\tau)$. This means that we
take the $x(\lambda)$ from \eqref{eq:SSC_worldline}
and determine the parameter $\lambda$
from \eqref{eq:pos_obs_image}, i.e.\ from 
$x(\lambda) \cdot u + \tau = 0$. Inserting the 
$\lambda = \lambda(u,\tau)$ so determined leads to  
\begin{defn}
	The \emph{SSC position observable} with 
	respect to $f$ is given by
	\begin{equation} \label{eq:SSC_pos_obs}
		\chi_\mu(u, \tau) = \frac{J_{\mu\rho} f^\rho}{f\cdot P} 
		+ \frac{\tau P_\mu}{(- u \cdot P)} 
		- \frac{J_{\lambda\rho} u^\lambda f^\rho}{(- f \cdot P)}
		\, \frac{P_\mu}{(- u \cdot P)} \; .
	\end{equation}
	Let us again stress the interpretation of this 
	expression: it is the SSC position with respect to 
	$f$ (i.e.\ a point on the `centre of energy' worldline 
	with respect to $f$) as localised on the hyperplane 
	characterised by unit normal $u$ and distance $\tau$ 
	to the origin, i.e.\ as seen in the Lorentz frame 
	with respect to $u$ at `time' $\tau$.
\end{defn}

Note that for this definition to make sense, $f$ 
does not have to be a fixed timelike future-directed 
vector: it can depend on the normal $u$ (and could 
even depend on $\tau$), and it can also depend on 
phase space\footnote
	{Various choices for $f$ were given 
	distinguished names in the literature. 
	The main ones, different from the Newton--Wigner 
	condition to be discussed here, are as follows. 
	If $f$ is meant to just characterise a fixed 
	`laboratory frame', which may be preferred for 
	any reason, like rotational symmetries in that 
	frame, the SSC is named after 
	Corinaldesi \& Papapetrou 
	\cite{Corinaldesi.Papapetrou:1951}. If $f$ is 
	proportional to the total linear momentum of the 
	system, the SSC is named after Tulczyjew \cite{Tulczyjew:1959} and Dixon
	\cite{Dixon:1970}. If $f$ is chosen in a somewhat 
	self-referential way to be the four-velocity of 
	the worldline that is to be determined by the 
	very SSC containing that $f$, the condition is 
	named after Frenkel  \cite{Frenkel:1926}, 
	Mathisson \cite{Mathisson:1937a,Mathisson:1937b}, 
	and Pirani \cite{Pirani:1956a,Pirani:1956b}.}.
Of course this means that according to this 
dependence of $f$, we will possibly be considering 
different worldlines for different choices of $u$.
\begin{exmp}
	\begin{enumerate}[label=(\roman*)]
		\item Choosing $f = u$, we are considering, for 
			each $u$, the SSC worldline with respect to $u$, 
			i.e.\ the \emph{centre of energy} worldline\footnote
				{Note that it was called `centre of mass' by Fleming \cite{Fleming:1965a}.}
			with respect to $u$. Using \eqref{eq:SSC_pos_obs}, 
			the centre of energy position observable has the form
			\begin{equation}
				\chi^\mathrm{CE}_\mu(u, \tau) = \frac{J_{\mu\rho} u^\rho}{u\cdot P} + \frac{\tau P_\mu}{(- u \cdot P)} \; .
			\end{equation}

		\item In the case of timelike four-momentum, we 
			can choose $f = P$ the four-momentum (the 
			Tulczyjew--Dixon SSC), such that the corresponding 
			SSC worldline is the centre of energy worldline in 
			the momentum rest frame of the system. This worldline, 
			which is obviously independent of $u$, was called the 
			\emph{centre of inertia} worldline by Fleming 
			\cite{Fleming:1965a}. The centre of inertia has the form
			\begin{equation}
				\chi^\mathrm{CI}_\mu(u, \tau) = -\frac{J_{\mu\rho} P^\rho}{m^2 c^2} + \frac{\tau P_\mu}{(- u \cdot P)} - \frac{J_{\lambda\rho} u^\lambda P^\rho}{m^2 c^2}\; \frac{P_\mu}{(- u \cdot P)} \; ,
			\end{equation}
			where $m = \sqrt{-P^2}/c$ is the mass of the system.
		\item Choosing $f = u + \frac{P}{mc}$ where 
			$m = \sqrt{-P^2}/c$ is the mass of the system (again 
			only possible in the case of timelike four-momentum), 
			we obtain the \emph{Newton--Wigner position observable}. 
			Evaluating \eqref{eq:SSC_pos_obs}, it has the form
			\begin{equation} \label{eq:NW_pos}
				\chi_\mu^\mathrm{NW}(u, \tau) = -\frac{J_{\mu\rho} \left(u^\rho + \frac{P^\rho}{mc}\right)}{mc - u\cdot P} + \frac{\tau P_\mu}{(- u \cdot P)} - \frac{J_{\lambda\rho} u^\lambda P^\rho}{mc (mc - u \cdot P)}\, \frac{P_\mu}{(- u \cdot P)} \; .
			\end{equation}
	\end{enumerate}
\end{exmp}

Of course, the SSC position observable 
\eqref{eq:SSC_pos_obs} will generally not be 
covariant in the sense of definition 
\ref{defn:pos_obs_cov} unless $f$ is also assumed 
to transform appropriately. If $f$ depends 
on $\Sigma\in\mathsf{SpHP}$ and $\gamma\in\Gamma$ 
and takes values in $V$ it seems obvious that for the 
resulting position to be covariant $f$ itself 
must be a covariant function under the 
combined actions on its domain and target spaces.
Indeed, we have
\begin{prop} \label{prop:SSC_pos_obs_cov}
	If the vector $f$ defining the SSC position 
	observable $\chi$ is a function  
	\begin{equation}
		f \colon \mathsf{SpHP} \times \Gamma \to V, \quad (\Sigma,\gamma) \mapsto f(\Sigma)(\gamma),
	\end{equation}
	such that  
	\begin{equation} \label{eq:SSC_vec_cov}
		f \Big((\Lambda, a) \cdot \Sigma\Big) \Big(\Phi_{(\Lambda, a)}(\gamma)\Big) = \Lambda \cdot \Big( f(\Sigma)(\gamma) \Big)
	\end{equation}
	for all $\Sigma \in \mathsf{SpHP}$, 
	$\gamma \in \Gamma$, and $(\Lambda, a) \in \mathcal P$,
	then $\chi$ is a covariant position observable. 
	Again we note that, just like in the transition 
	from \eqref{eq:pos_obs_cov_1} to \eqref{eq:pos_obs_cov_2}, 
	we may rewrite \eqref{eq:SSC_vec_cov} equivalently 
	as expressing the invariance of $f$ (i.e.\ its graph) 
	under simultaneous actions on its domain and target 
	spaces (using that translations act trivially on the 
	target space $V$):
	\begin{equation}
		f = \Lambda \circ f \circ \left( (\Lambda,a)^{-1} \times \Phi_{(\Lambda,a)^{-1}} \right)
	\end{equation}

	\begin{proof}
		At first, suppose we are given a future-directed 
		timelike four-momentum $P \in V$ and an angular 
		momentum tensor $J \in \bigwedge^2 V^*$, as well 
		as a future-directed timelike vector $f$ for the 
		definition of an SSC. In addition, fix a Poincaré 
		transformation $(\Lambda, a) \in \mathcal P$. 
		If we now consider (a) the SSC worldline for $P$ 
		and $J$ with respect to $f$, and (b) the SSC 
		worldline for the transformed four-momentum 
		$P'=\Lambda P$ and angular momentum  
		$J' = ((\Lambda^{-1})^\top \otimes (\Lambda^{-1})^\top)J + a^\flat \wedge (\Lambda^{-1})^\top P^\flat$ 
		(compare \eqref{eq:coadjoint_rep_components_b}) 
		with respect to the transformed vector $\Lambda f$, 
		it is easy to check that the second worldline is 
		the Poincaré transform by $(\Lambda, a)$ of the 
		first. That is, by Poincaré transforming the 
		four-momentum and angular momentum of the system 
		as well as the `direction vector' for the SSC, we 
		Poincaré transform the SSC worldline.

		Now, the SSC position $\chi(\Sigma)(\gamma)$ is 
		defined to be the intersection of the hyperplane 
		$\Sigma$ with the SSC worldline of $\gamma$ with 
		respect to $f(\Sigma)(\gamma)$. Thus, the `new position'
		\begin{equation}
			\chi \Big((\Lambda, a) \cdot \Sigma\Big) \Big(\Phi_{(\Lambda, a)}(\gamma)\Big)
		\end{equation}
		is the intersection of the transformed hyperplane 
		$(\Lambda, a) \cdot \Sigma$ with the SSC worldline 
		of the transformed system $\Phi_{(\Lambda, a)}(\gamma)$ 
		with respect to the transformed vector 
		$\Lambda \cdot \Big( f(\Sigma)(\gamma) \Big)$, 
		where we used the covariance requirement 
		\eqref{eq:SSC_vec_cov}. But according to our earlier 
		considerations, this means that the `new position' 
		is the intersection of the transformed hyperplane 
		with the transform of the original SSC worldline~-- 
		i.e.\ the transform of the original position 
		$\chi(\Sigma)(\gamma)$. This means that the position 
		observable is covariant.
	\end{proof}
\end{prop}

Since the vectors defining the centre of energy, 
the centre of inertia and the Newton--Wigner position 
satisfy \eqref{eq:SSC_vec_cov}, all of these are 
covariant position observables. We stress once 
more that for this to be true we need to take into 
account the action of the Poincaré group on 
$\mathsf{SpHP}$. This remark is particularly 
relevant in the Newton--Wigner case, in which $f$ 
is the sum of two vectors, $u$ and $P/(mc)$, the 
first being associated to an element of 
$\mathsf{SpHP}$ and the second to an element of 
$\Gamma$. Covariance cannot be expected to hold 
for non-trivial actions on $\Gamma$ alone. In the 
next section we will offer an insight as to why 
this somewhat `hybrid' combination for $f$ in 
terms of an `external' vector $u$ and an 
`internal' vector $P/(mc)$ appears. The latter is 
internal, or dynamical, in the sense that it is defined 
entirely by the physical state of the system, i.e.\ 
a point in $\Gamma$, while the former is external, or 
kinematical, in the sense that it refers to the choice 
of $\Sigma\in\mathsf{SpHP}$, which is entirely 
independent of the physical system and its state.

Finally, we will need the following well-known 
result for SSCs with respect to different vectors 
$f$, which was first shown by M{\o}ller in 1949 in 
\cite{Moeller:1949}; see also 
\cite[theorem 17]{Giulini:2015} for a recent 
and more geometric discussion:
\begin{thm}[Møller disc and radius] \label{thm:Moller_disc}
	Suppose we are given the future-directed timelike 
	four-momentum vector $P \in V$ and the angular momentum 
	tensor $J \in \bigwedge^2 V^*$ of some physical system. 
	Consider the bundle of all possible SSC worldlines 
	\eqref{eq:SSC_worldline} for this system, defined by 
	considering all future-directed timelike vectors $f$. 
	The intersection of this bundle with any hyperplane 
	$\Sigma \in \mathsf{SpHP}$ orthogonal to $P$ is a 
	two-dimensional disc (the so-called \emph{M\o ller disc}) 
	in the plane orthogonal to the Pauli--Lubański vector 
	$W = (*(P^\flat \wedge J))^\sharp$, whose centre is the 
	centre of inertia on $\Sigma$ and whose radius is the 
	\emph{Møller radius}
	\begin{equation}
		R_M = \frac{S}{mc} \; ,
	\end{equation}
	where $S = \sqrt{W^2}/(mc)$ is the spin of the system 
	and $m = \sqrt{-P^2}/c$ its mass.
\end{thm}

\subsection{The centre of spin condition}

For a system with timelike four-momentum, the 
Pauli--Lubański vector $W$ has the interpretation 
of being ($mc$ times) the spin vector in the momentum 
rest frame. We now define the spin vector in an 
arbitrary Lorentz frame by boosting $W/(mc)$ to the 
new frame:
\begin{defn}
	Given the timelike four-momentum $P \in V$ and the 
	Pauli--Lubański vector $W \in P^\perp$ of a 
	physical system, its \emph{spin vector} in the 
	Lorentz frame given by the future-directed unit 
	timelike vector $u$ is
	\begin{equation}
		s(u) := B(u) \cdot \frac{W}{mc} \; ,
	\end{equation}
	where $B(u) \in \mathcal L_+^\uparrow$ is the 
	unique Lorentz boost with respect to $\frac{P}{mc}$ 
	(i.e.\ containing $\frac{P}{mc}$ in its timelike 
	2-plane of action) that maps $\frac{P}{mc}$ to 
	$u$, with $m = \sqrt{-P^2}/c$ being the mass. In terms 
	of components, this boost is given by\footnote
		{Generally, given two unit 
		timelike future-pointing vectors $n_1$ and $n_2$,
		then the boost that maps $n_1$ onto $n_2$ and 
		fixes the spacelike plane orthogonal to 
		$\mathrm{span}\{n_1,n_2\}$ is given by the 
		combination $\rho_{n_1+n_2}\circ\rho_{n_1}$ of two 
		hyperplane-reflections, where 
		$\rho_n := \mathrm{id}_V - 2\frac{n\otimes n^\flat}{n^2}$ 
		is the reflection at the hyperplane orthogonal to 
		$n$. Setting $n_1 = P/(mc)$ and $n_2 = u$ gives 
		\eqref{eq:linking_boost_explicit}.}
	\begin{equation} \label{eq:linking_boost_explicit}
		B^\mu_{\hphantom{\mu}\nu}(u)
		= \delta^\mu_\nu + \frac{\left(\frac{P^\mu}{mc} 
		+ u^\mu\right) \left(\frac{P_\nu}{mc} 
		+ u_\nu\right)}{1 - u \cdot \frac{P}{mc}} 
		- 2 \frac{u^\mu P_\nu}{mc} \; .
	\end{equation}
\end{defn}

\begin{defn}
	A \emph{centre of spin} position observable for 
	a 	classical Poincaré-invariant system 
	$(\Gamma, \omega, \Phi)$ with timelike four-momentum 
	is a position observable $\chi$ satisfying
	\begin{equation}
		s_\mu(u) = - \frac{1}{2} \varepsilon_{\mu\nu\rho\sigma} u^\nu S^{\rho\sigma}(u),
	\end{equation}
	where $S_{\rho\sigma}(u) := J_{\mu\nu} - \chi_\mu(u, \tau) P_\nu + \chi_\nu(u, \tau) P_\mu$ 
	is the spin tensor\footnote
		{Since $\frac{\partial\chi(u,\tau)}{\partial\tau}$ is 
		proportional to $P$, the spin tensor is independent 
		of $\tau$.}
	with respect to $\chi$. Expressed in terms of the 
	Hodge operator, this condition reads
	\begin{equation}
		s(u) = (*(u^\flat \wedge S(u)))^\sharp.
	\end{equation}
\end{defn}
With respect to an orthonormal basis 
$\{u = e_0, \dots, e_3\}$ adapted to $u$, the 
centre of spin condition takes the form
\begin{equation}
	s_0(u) = 0, \quad
	s_a(u) = - \frac{1}{2} \varepsilon_{a0\rho\sigma} S^{\rho\sigma}(u) 
		= \frac{1}{2} {^{(3)}\varepsilon}_{abc} S^{bc}(u),
\end{equation}
through which it acquires an immediate interpretation: 
a position observable is a centre of spin if and only 
if, for any Lorentz frame $u$, the spin vector 
defined by boosting the Pauli--Lubański vector 
to $u$ really generates spatial rotations around 
the point given by the position observable.

We will now rewrite the centre of spin condition. 
Since $S(u) = J - (\chi(u,\tau))^\flat \wedge P^\flat$, 
we can rewrite the Pauli--Lubański vector as 
$W = \left[* \left(\frac{P^\flat}{mc} \wedge J \right) \right]^\sharp = \left[* \left(\frac{P^\flat}{mc} \wedge S(u) \right) \right]^\sharp$. 
Thus, the centre of spin condition takes the form
\begin{equation}
	(B(u)^{-1})^\top \left[*\left(\frac{P^\flat}{mc} \wedge S(u)\right)\right] = *(u^\flat \wedge S(u)).
\end{equation}
Since $B(u)$ is a Lorentz transformation, i.e.\ an 
isometry of $(V,\eta)$, and it maps $P/(mc)$ to $u$, 
this is equivalent to
\begin{equation}
	u^\flat \wedge \left((B(u)^{-1})^\top \otimes (B(u)^{-1})^\top\right) (S(u)) = u^\flat \wedge S(u).
\end{equation}
Using the explicit form \eqref{eq:linking_boost_explicit} 
of $B(u)$, we see that
\begin{equation}
	\left((B(u)^{-1})^\top \otimes (B(u)^{-1})^\top\right) (S(u)) = S(u) + \frac{\frac{P^\flat}{mc} \wedge \left(\iota_{u + \frac{P}{mc}} S(u)\right)}{1 - u \cdot \frac{P}{mc}} + u^\flat \wedge (\ldots).
\end{equation}
Thus, we have the following:
\begin{lemqed}
	The centre of spin condition is equivalent to
	\begin{equation} \label{eq:centre_of_spin_equiv_1}
		u^\flat \wedge P^\flat \wedge \left(\iota_{u + \frac{P}{mc}} S(u)\right) = 0. \qedhere
	\end{equation}
\end{lemqed}
Since the Newton--Wigner position observable is defined 
by the SSC $\iota_{u + \frac{P}{mc}} S(u) = 0$, the 
preceding result immediately implies
\begin{thmqed}
	The Newton--Wigner position observable $\chi^\mathrm{NW}$ 
	is a centre of spin.
\end{thmqed}

Further rewriting the centre of spin condition, we see 
that \eqref{eq:centre_of_spin_equiv_1} is equivalent to
\begin{equation}
	\iota_{u + \frac{P}{mc}} S(u) \in \mathrm{span} \{u^\flat, P^\flat\}.
\end{equation}
Due to the antisymmetry of $S(u)$, this is equivalent to
\begin{equation} \label{eq:centre_of_spin_equiv_2}
	\iota_{u + \frac{P}{mc}} S(u) \in \mathrm{span} \left\{u^\flat - \frac{P^\flat}{mc} \right\}.
\end{equation}
Using this, we can show:
\begin{lem} \label{lem:centre_of_spin_diff_NW}
	$\chi$ is a centre of spin $\iff$ 
	$\chi(u, \tau) - \chi^\mathrm{NW}(u, \tau) \in \mathrm{span}\{u, P\}$.

	\begin{proof}
		Writing $D := \chi(u, \tau) - \chi^\mathrm{NW}(u,\tau)$, 
		the spin tensor of $\chi$ may be expressed as $S(u) = S^\mathrm{NW}(u) - D^\flat \wedge P^\flat$. 
		Thus, \eqref{eq:centre_of_spin_equiv_2} is equivalent to
		\begin{equation} \label{eq:centre_of_spin_equiv_3}
			\iota_{u + \frac{P}{mc}} (D^\flat \wedge P^\flat) \in \mathrm{span} \left\{u^\flat - \frac{P^\flat}{mc} \right\}.
		\end{equation}
		We have $\iota_{u + \frac{P}{mc}} (D^\flat \wedge P^\flat) = (D\cdot u + \frac{D\cdot P}{mc}) P^\flat - (P\cdot u - mc) D^\flat$, 
		and thus \eqref{eq:centre_of_spin_equiv_3} implies that 
		for all $v \in u^\perp \cap P^\perp$, we have
		\begin{equation}
			v\cdot D = 0.
		\end{equation}
		But this means $D \in (u^\perp \cap P^\perp)^\perp = \mathrm{span}\{u, P\}$.

		Conversely, if $D \in \mathrm{span}\{u, P\}$, we 
		have $\iota_{u + \frac{P}{mc}} (D^\flat \wedge P^\flat) \in \mathrm{span} \left\{ \iota_{u + \frac{P}{mc}} (u^\flat \wedge P^\flat) \right\}$. 
		But now
		\begin{equation}
			\iota_{u + \frac{P}{mc}} (u^\flat \wedge P^\flat) = \left(-1 + u\cdot \frac{P}{mc}\right) P^\flat - (u\cdot P - mc) u^\flat = (mc - u\cdot P) \left(u^\flat - \frac{P^\flat}{mc}\right),
		\end{equation}
		and thus we have \eqref{eq:centre_of_spin_equiv_3}, 
		i.e.\ $\chi$ is a centre of spin.
	\end{proof}
\end{lem}

We can now prove the main result of this section.
\begin{thm} \label{thm:NW_SSC_centre_of_spin}
	The Newton--Wigner position observable $\chi^\mathrm{NW}$ 
	is the only centre of spin position observable that is 
	continuous and defined by an SSC.
	
	\begin{proof}
		Let $\chi$ be an SSC position observable. Writing 
		$D(u,\tau) := \chi(u, \tau) - \chi^\mathrm{NW}(u,\tau)$, 
		we know by the M\o ller disc theorem (theorem 
		\ref{thm:Moller_disc}) that the projection of 
		$D(u, \tau)$ orthogonal to $P$ is orthogonal to the 
		Pauli--Lubański vector $W$. Thus, since $P$ itself 
		is orthogonal to $W$, we have
		\begin{equation} \label{eq:NW_SSC_unique_proof_orth}
			D(u, \tau) \perp W
		\end{equation}
		for any $(u, \tau) \in \mathsf{SpHP}$. In addition, 
		we know that $D(u, \tau) \perp u$; in particular, 
		$D(u, \tau)$ is spacelike for any 
		$(u, \tau) \in \mathsf{SpHP}$.

		Now suppose that $\chi$ is a centre of spin. By lemma 
		\ref{lem:centre_of_spin_diff_NW} this means that
		\begin{equation}
			D(u, \tau) \in \mathrm{span}\{u, P\}
		\end{equation}
		for all $(u, \tau) \in \mathsf{SpHP}$. Using 
		\eqref{eq:NW_SSC_unique_proof_orth} and $P \perp W$, 
		we conclude that
		\begin{equation}
			\text{for all} \; u \; \text{with} \; u\cdot W \ne 0 : D(u, \tau) \in \mathrm{span}\{P\}.
		\end{equation}
		Since $D(u, \tau)$ has to be spacelike, we thus have shown
		\begin{equation}
			D(u, \tau) = 0 \; \text{for all} \; u \; \text{with} \; u\cdot W \ne 0.
		\end{equation}
		If $W \ne 0$, the set of future-directed unit timelike 
		$u$ satisfying $u \cdot W \ne 0$ is dense in the 
		hyperboloid of all possible $u$, and thus assuming 
		continuity of $\chi$, we conclude that $D(u, \tau) = 0$ 
		for all $u$, finishing the proof.

		If $W = 0$, then by the M\o ller disc theorem all 
		SSC worldlines coincide, and thus we also have 
		$\chi = \chi^\mathrm{NW}$.
	\end{proof}
\end{thm}

Looking back into the various steps of the proofs 
it is interesting to note how the 
`extrinsic--intrinsic' combination $u+P/(mc)$ 
for $f$ came about. It entered through 
the unique boost transformation 
\eqref{eq:linking_boost_explicit} that 
was needed in order to transform an intrinsic 
quantity to an externally specified rest frame. 
The intrinsic quantity is the spin vector 
in the momentum rest frame, i.e.\ the
Pauli--Lubański vector, which is a function 
of $\Gamma$ only, and the externally specified 
frame is defined by $u$, which is independent 
of $\Gamma$ and determined through the choice of 
$\Sigma\in\mathsf{SpHP}$.

\section{A Newton--Wigner theorem for classical elementary systems}
\label{sec:NW_theorem}

For elementary Poincaré-invariant quantum systems 
--~i.e.\ quantum systems with an irreducible unitary 
action of the Poincaré group~-- the Newton--Wigner 
position operator is uniquely characterised by 
transforming `as a position should' under translations, 
rotations and time reversal, having commuting components 
and satisfying a regularity condition. This has been 
well-known since the original publication by Newton 
and Wigner \cite{Newton.Wigner:1949}. As advertised in 
the introduction, we shall now prove an analogous 
statement for classical systems.

For the whole of this section, we fix a future-directed 
unit timelike vector $u$ defining a Lorentz frame, 
and an adapted positively oriented orthonormal basis 
$\{u = e_0, \dots, e_3\}$. Unless otherwise stated, 
phrases such as `temporal', `spatial' and the like 
refer to the preferred time direction given by $u$. 
We will raise and lower spatial indices by the Euclidean 
metric $\delta$ induced by the Minkowski metric $\eta$ 
on the orthogonal complement of $u$; the components of 
$\delta$ in the adapted basis are simply given by the 
usual Kronecker delta. We denote the spatial volume form 
by $^{(3)}\varepsilon = \iota_u\varepsilon$.

We will employ a `three-vector' notation for spatial 
vectors, for example writing $\vec A = (A^a)$. We then use 
the usual three-vector notations for the Euclidean scalar 
product $\vec A \cdot \vec B = A_a B^a$, the Euclidean 
norm $|\vec A| := \sqrt{\vec A^2}$ and the vector product 
$(\vec A \times \vec B)_a = {^{(3)}\varepsilon_{abc}} A^b B^c$.

\subsection{Classical elementary systems}

In the quantum case, an elementary system is given by a 
Hilbert space with an \emph{irreducible} unitary action 
of the Poincaré group~-- i.e.\ each state of the system 
is connected to any other by a Poincaré transformation. 
In direct analogy, we define the notion of a classical 
elementary system:
\begin{defn}
	A \emph{classical elementary system} is a classical 
	Poincaré-invariant system $(\Gamma, \omega, \Phi)$, 
	where $\Phi$ is a transitive action of the proper 
	orthochronous Poincaré group $\mathcal P_+^\uparrow$.
\end{defn}
Note the we only assumed an action of the identity 
connected component of the Poincaré group, whereas  
Arens in \cite{Arens:1971b} considered the whole 
Poincaré group. In the classical context, simple 
transitivity replaces irreducibility in the quantum 
case.  

Arens classified the classical elementary systems\footnote
	{In fact, Arens classified what he called 
	\emph{one-particle} elementary systems (systems 
	that admit a map from $\Gamma$ to the set of lines 
	in Minkowski space which is equivariant with respect 
	to a certain subgroup of $\mathcal P_+^\uparrow$). 
	However, he also proved that this `one-particle' 
	condition is fulfilled for an elementary system 
	if and only if the four-momentum is not zero.}
in \cite{Arens:1971b}; the classification proceeds 
in terms of the system's four-momentum and 
Pauli--Lubański vector (similar to the Wigner
classification in the quantum case \cite{Wigner:1939}). We are only 
interested in the case of timelike four-momentum. 
For this case, the phase space can be explicitly 
constructed as follows:

\begin{thm}[Phase space of a classical elementary system] \label{thm:phase_space}
	Any classical elementary system with timelike 
	four-momentum is equivalent (in the sense of a 
	symplectic isomorphism respecting the action of 
	$\mathcal P_+^\uparrow$) to precisely one of the 
	following two cases:
	\begin{enumerate}[label=(\roman*)]
		\item (Spin zero, one parameter $m \in \mathbb R_+$)
			\begin{itemize}
				\item Phase space $\Gamma = T^*\mathbb R^3$ with 
					coordinates $(\vec x, \vec p)$, symplectic 
					form $\omega = \D x^a \wedge \D p_a$
				\item Poincaré generators (i.e.\ component functions of
					the momentum map):
					\begin{subequations}
					\begin{align}
						\text{spatial translations} \quad P_a &= p_a \\
						\text{time translation} \quad P_0 &= -\sqrt{m^2 c^2 + \vec p^2} \\
						\text{rotations} \quad J_{ab} &= x_a p_b - x_b p_a \\
						\text{boosts} \quad J_{a0} &= P_0 x_a
					\end{align}
					\end{subequations}
			\end{itemize}

		\item (Spin non-zero, two parameters $m, S \in \mathbb R_+$)
			\begin{itemize}
				\item Phase space $\Gamma = T^*\mathbb R^3 \times \mathsf S^2$ 
					with coordinates $(\vec x, \vec p)$ for $T^*\mathbb R^3$, 
					symplectic form $\omega = \D x^a \wedge \D p_a + S \cdot \D\Omega^2$ 
					where $\D\Omega^2$ is the standard volume form 
					on $\mathsf S^2$. We denote the phase space function 
					projecting onto the second factor $\mathsf S^2$ by 
					$\hat s \colon \Gamma \to \mathsf S^2 \subset\mathbb R^3$. 
					The \emph{spin vector} observable is the 
					$\mathsf S^2_S$-valued phase space function 
					$\vec s := S \cdot \hat s$; its components satisfy 
					the Poisson bracket relations
					\begin{equation}
						\{s_a, s_b\} = {^{(3)}\varepsilon_{abc}} s^c.
					\end{equation}
					Here $\mathsf S^2_S \subset \mathbb R^3$ denotes 
					the 2-sphere of radius $S$ in $\mathbb R^3$.
				\item Poincaré generators (i.e.\ component functions of
					the momentum map):
					\begin{subequations}
					\begin{align}
						\text{spatial translations} \quad P_a &= p_a \\
						\text{time translation} \quad P_0 &= -\sqrt{m^2 c^2 + \vec p^2} \\
						\text{rotations} \quad J_{ab} &= x_a p_b - x_b p_a + {^{(3)}\varepsilon_{abc}} s^c \\
						\text{boosts} \quad J_{a0} &= P_0 x_a - \frac{(\vec p \times \vec s)_a}{mc - P_0}
					\end{align}
					\end{subequations}
			\end{itemize}
	\end{enumerate}
\end{thm}

Note that in fact the explicit construction of the systems 
in \cite{Arens:1971b} as co-adjoint orbits of $\mathcal P_+^\uparrow$ 
is quite different in appearance to the forms given above. 
However, one can show that the above systems are indeed 
elementary systems (i.e.\ that the action of $\mathcal P_+^\uparrow$ 
is transitive), and thus due to Arens' uniqueness result 
they are possible representatives of their respective 
classes. We will use the forms given above, which were 
anticipated by Bacry in \cite{Bacry:1967}, since they will 
be easier to explicitly work with. To unify notation, we 
let $S = 0, \vec s := 0$ in the case of zero-spin systems. 
Furthermore, we introduce the open subset of phase space 
$\Gamma^* := \Gamma \setminus \{|\vec P| = 0\}$ and the 
$\mathsf S^2$-valued function $\hat P := \frac{\vec P}{|\vec P|}$ 
on $\Gamma^*$.

Using the explicit form of the systems given in theorem 
\ref{thm:phase_space}, one directly checks:
\begin{lemqed} \label{lem:mom_spin_CIV}
	For a classical elementary system with timelike 
	four-momentum, the functions $P_a, \hat P \cdot \vec s$ 
	(or just the $P_a$ in the case of zero spin) form a 
	complete involutive set on $\Gamma^*$ (or the whole of 
	$\Gamma$ in the case of zero spin).
\end{lemqed}

The behaviour of the momentum and spin vectors under translations 
and rotations is also easily obtained:
\begin{lem} \label{lem:mom_spin_trafo}
	For a classical elementary system with timelike 
	four-momentum, $\vec P$ and $\vec s$ are invariant under 
	translations and `transform as vectors' under spatial 
	rotations, i.e.\ we have
	\begin{equation}
		\{P_a, V_b\} = 0, \quad 
		\{J_{ab}, V_c\} = \delta_{ac} V_b - \delta_{bc} V_a 
		\quad \text{for}\quad 
		\vec V = \vec P, \vec s.
	\end{equation}
	\begin{proof}
		For $\vec P$, these are part of the Poincaré algebra relations and thus true by definition. For $\vec s$, they are easily confirmed using the explicit form of the Poincaré generators.
	\end{proof}
\end{lem}

For our considerations, we will need to know how the 
\emph{time reversal} operation with respect to the 
hyperplane in $M$ through the origin $o \in M$ and 
orthogonal to $u = e_0$ is implemented on phase 
space. In order to get this right, we recall that 
the incorporation of time reversal in the context 
of Special Relativity corresponds, by its very 
definition, to a particular upward $\mathbb{Z}_2$ 
extension\footnote
	{Here we are using the terminology of 
	\cite[p.\,xx]{Conway.EtAl:AOFG}, according to which 
	a group $G$ with normal subgroup $A$ and quotient 
	$G/A \cong B$ is either called an 
	\emph{upward extension} of $A$ by $B$ or a 
	\emph{downward extension} of $B$ by $A$.}
of $\mathcal{P}_+^\uparrow$, i.e.\ the formation 
of a new group called 
$\mathcal{P}_+^\uparrow\cup\mathcal{P}_-^\downarrow$ 
of which $\mathcal{P}_+^\uparrow$ is a normal 
subgroup with $(\mathcal{P}_+^\uparrow\cup\mathcal{P}_-^\downarrow)/\mathcal{P}_+^\uparrow\cong\mathbb{Z}_2$. 
It is the particular nature of this extension that 
eventually defines what is meant by 
\emph{time reversal}: it consists in the requirement 
that the outer automorphism induced by the only 
non-trivial element of $\mathbb{Z}_2$ on the Lie 
algebra $\mathfrak{p}$ of $\mathcal{P}_+^\uparrow$ 
shall be the one which reverses the sign of spatial 
translations and rotations and leaves invariant 
boosts and time translations; see, e.g., 
\cite{Bacry.LevyLeblond:1968}. Implementing time 
reversal on phase space then means to extend the 
action of $\mathcal{P}_+^\uparrow$ to an action 
of $\mathcal{P}_+^\uparrow\cup\mathcal{P}_-^\downarrow$.

Now, according to this scheme, we can immediately 
write down how our particular time reversal transformation 
on phase space, $T_u\colon\Gamma\to\Gamma$, acts on the 
Poincaré generators, i.e.\ the component functions of the 
momentum map: 
\begin{equation}
	P_a \circ T_u = - P_a \; , \quad 
	J_{ab} \circ T_u = - J_{ab} \; , \quad 
	J_{a0} \circ T_u = J_{a0} \; ,  \quad 
	P_0 \circ T_u = P_0
\end{equation}
From this the well-known result follows 
that time reversal (as defined above) necessarily 
corresponds to an \emph{anti}-symplectomorphism. 
Hence, in the process of extending our symplectic 
action of $\mathcal{P}_+^\uparrow$ on $\Gamma$ to an action 
of $\mathcal{P}_+^\uparrow\cup\mathcal{P}_-^\downarrow$ 
satisfying the time reversal criterion above, we had to 
generalise to possibly \emph{anti}-symplectomorphic 
actions. This is akin to the situation in Quantum 
Mechanics, where, as is well-known, time reversal 
necessarily corresponds to an \emph{anti}-unitary 
transformation.

It is now clear how time reversal is implemented 
in the case at hand:
\begin{lemqed} \label{lem:time_rev}
	For an elementary system as in theorem
	\ref{thm:phase_space}, time reversal with respect 
	to the hyperplane through the origin and orthogonal 
	to $u = e_0$ is given by
	\begin{equation}
			T_u \colon (\vec x, \vec p, \hat s) \mapsto (\vec x, - \vec p, - \hat s). \qedhere
		\end{equation}
\end{lemqed}
Unless otherwise stated, in the following we will 
always mean time reversal with respect to the 
hyperplane through the origin and orthogonal to 
$u = e_0$ when saying `time reversal'.

\subsection{Statement and interpretation of the Newton--Wigner theorem}

The classical Newton--Wigner theorem we are going to 
prove can be formulated very similar to the quantum case:
\begin{thm}[Classical Newton--Wigner theorem] \label{thm:NW}
	For a classical elementary system with timelike 
	four-momentum, there is a unique $\mathbb R^3$-valued 
	phase space function $\vec X$ that
	\begin{enumerate}[label=(\roman*)]
		\item is $C^1$,
		\item has Poisson-commuting components,
		\item \label{ass:NW_transl} satisfies the canonical 
			Poisson relations $\{X^a, P_b\} = \delta^a_b$ 
			with the generators of spatial translations with 
			respect to $u = e_0$,
		\item \label{ass:NW_rot} transforms `as a (position) 
			vector' under spatial rotations with respect to 
			$u = e_0$, i.e.\ satisfies $\{J_{ab}, X^c\} = \delta_a^c X_b - \delta_b^c X_a$, and
		\item is invariant under time reversal with respect to 
			the hyperplane through the origin and orthogonal to 
			$u = e_0$, i.e.\ satisfies $\vec X \circ T_u = \vec X$.
	\end{enumerate}

	In terms of the Poincaré generators, it is given by
	\begin{equation} \label{eq:NW_pos_thm}
		X_a = -\frac{J_{a0}}{mc} - \frac{J_{ab} P^b}{mc (mc - P_0)} - \frac{J_{b0} P^b}{P_0 mc (mc - P_0)} P_a \; ,
	\end{equation}
	where $m = \sqrt{P_0^2 - \vec P^2}/c$ is the mass of the system.
\end{thm}
Before proving the theorem in the next section, 
we will now discuss the interpretation of the 
`position' $\vec X$ it characterises. We want to 
interpret the value of $\vec X$ (in some state 
$\gamma \in \Gamma$) as the spatial components 
of a point in Minkowski spacetime $M$. Since $\vec X$ 
is invariant under time reversal with respect to 
the hyperplane through the origin and orthogonal 
to $u = e_0$, it can be interpreted as defining a 
point on this hyperplane. Thus, if we want to use 
the phase space function from the Newton--Wigner 
theorem to define a position observable $\chi$ in 
the sense of section \ref{sec:pos_obs}, we should 
set (in our basis adapted to $u$)
\begin{equation} \label{eq:pos_obs_def_by_NW_thm}
	\chi^a(u, \tau = 0) := X^a \; , \quad
	\chi^0(u, \tau = 0) := 0.
\end{equation}
The transformation behaviour of $\vec X$ under spatial 
translations and rotations (i.e.\ assumptions 
\ref{ass:NW_transl} and \ref{ass:NW_rot} of theorem 
\ref{thm:NW}) will then ensure that the position 
observable $\chi$ be covariant (in the sense of definition 
\ref{defn:pos_obs_cov}) regarding these transformations.

In fact, comparing \eqref{eq:NW_pos_thm} to the expression 
\eqref{eq:NW_pos} for the Newton--Wigner position 
observable $\chi^\mathrm{NW}$, we see that we have 
(in our adapted basis)
\begin{equation}
	\chi^{\mathrm{NW},a}(u, \tau = 0) = X^a \; , \quad
	\chi^{\mathrm{NW},0}(u, \tau = 0) = 0\colon
\end{equation}
the position $\vec X$ characterised by theorem 
\ref{thm:NW} is the one given by the Newton--Wigner 
position observable $\chi^\mathrm{NW}$ on the 
hyperplane $(u, 0) \in \mathsf{SpHP}$ (which is a 
covariant position observable due to proposition 
\ref{prop:SSC_pos_obs_cov}). Let us also remark that 
since any position observable's dependence on $\tau$ 
is fixed by \eqref{eq:pos_obs_time_der}, a position 
observable satisfying \eqref{eq:pos_obs_def_by_NW_thm} 
is equal to the Newton--Wigner observable 
$\chi^\mathrm{NW}$ on the whole family of hyperplanes 
$\Sigma \in \mathsf{SpHP}$ with normal vector $u$.

Combining this identification with the observation that 
we can freely choose the origin $o\in M$, we can restate 
the Newton--Wigner theorem in the following form:
\begin{thmqed}[Classical Newton--Wigner theorem, version 2] \label{thm:NW2}
	For a classical elementary system with timelike 
	four-momentum, given any hyperplane $\Sigma = (u, \tau) \in \mathsf{SpHP}$, 
	there is a unique $\Sigma$-valued phase space function 
	$\chi^\mathrm{NW}(\Sigma)$ that
	\begin{subequations}
	\begin{enumerate}[label=(\roman*)]
		\item is $C^1$,
		\item has Poisson-commuting components, i.e.
			\begin{equation}
				\left\{ \chi^{\mathrm{NW},\mu}(\Sigma), \chi^{\mathrm{NW},\nu}(\Sigma) \right\} = 0,
			\end{equation}
		\item satisfies the canonical Poisson relations with 
			the generators of spatial translations with respect 
			to $u$, i.e.
			\begin{equation}
				v_\mu w^\nu \left\{ \chi^{\mathrm{NW},\mu}(\Sigma), P_\nu \right\} = v\cdot w \; \text{for} \; v, w \in u^\perp,
			\end{equation}
		\item transforms `as a position' under spatial rotations 
			with respect to $u$, i.e.\ satisfies
			\begin{equation}
				v^\mu \tilde v^\nu w_\rho \left\{ J_{\mu\nu}, \chi^{\mathrm{NW},\rho}(\Sigma) \right\} = v^\mu \tilde v^\nu w_\rho \left[\delta_\mu^\rho \chi^\mathrm{NW}_\nu(\Sigma) - \delta_\nu^\rho \chi^\mathrm{NW}_\mu(\Sigma)\right] \; \text{for} \; v, \tilde v, w \in u^\perp,
			\end{equation}
			and
		\item is invariant under time reversal with respect 
			to $\Sigma$.
	\end{enumerate}
	\end{subequations}

	These $\chi^\mathrm{NW}(\Sigma)$ together form the 
	Newton--Wigner observable as given by \eqref{eq:NW_pos}.
\end{thmqed}

\subsection{Proof of the Newton--Wigner theorem}

\begin{proof}[Proof of theorem \ref{thm:NW}]
For the whole of the proof, we will work with the 
explicit form of the phase space of our elementary 
system given in theorem \ref{thm:phase_space}. It is 
easily verified that in this explicit form, $\vec x$ 
(i.e.\ the coordinate of the base point in 
$T^*\mathbb R^3$) is a phase space function with the 
properties demanded for $\vec X$. Thus we need to 
prove uniqueness. Our proof will follow the proof of 
the quantum-mechanical Newton--Wigner theorem given 
by Jordan in \cite{Jordan:1980}, some parts of which 
can be applied literally to the classical case.

We will several times need the following.
\begin{lem} \label{lem:invar_function}
	Consider a classical elementary system with timelike 
	four-momentum, with phase space $\Gamma$, and some 
	open subset $\tilde\Gamma$ of 
	$\Gamma^* = \Gamma \setminus \{|\vec P| = 0\}$. 
	Let $f$ be an $\mathbb R$-valued $C^1$ function 
	defined on $\tilde\Gamma$ that is invariant under 
	spatial translations and rotations, i.e.\ 
	$\{P_a, f\} = 0 = \{J_{ab}, f\}$. Then $f$ is a function of 
	$|\vec P|, \hat P \cdot \vec s$. \footnote
		{By `$f$ is a function of $|\vec P|, \hat P \cdot \vec s$ ' 
		we mean that $f$ depends on phase space only via 
		$|\vec P|, \hat P \cdot \vec s$, i.e.\ that there 
		is a $C^1$ function $F\colon U \to \mathbb R$, 
		$U = \left\{(|\vec P|(\gamma), (\hat P \cdot \vec s)(\gamma)) : \gamma \in \tilde\Gamma\right\} \subset \mathbb R_+ \times [-S,S]$ 
		satisfying
		\[f(\gamma) = F(|\vec P|(\gamma), (\hat P \cdot \vec s)(\gamma)) \; \text{for all} \; \gamma \in \tilde\Gamma.\]}

	\begin{proof}
		$f$ Poisson-commutes with $\vec P$ and $J_{ab}$. 
		Therefore it also Poisson-commutes with $\vec P$ 
		and $\frac{1}{2} {^{(3)}\varepsilon^{abc}} \hat P_a J_{bc} = \hat P \cdot \vec s$. 
		Now $\vec P, \hat P \cdot \vec s$ form a complete 
		involutive set on $\Gamma^*$ (lemma 
		\ref{lem:mom_spin_CIV}), so since $f$ Poisson-commutes 
		with them, it must be a function of 
		$\vec P, \hat P \cdot \vec s$. Since $f$ and 
		$\hat P \cdot \vec s$ are rotation invariant (by 
		lemma \ref{lem:mom_spin_trafo}), $f$ must be a 
		function of $|\vec P|, \hat P \cdot \vec s$.
	\end{proof}
\end{lem}

Let now $\vec X$ be an observable as in the 
statement of theorem \ref{thm:NW}, and consider 
the difference $\vec d := \vec X - \vec x$. Due to 
the assumptions of theorem \ref{thm:NW}, $\vec d$ 
is $C^1$, is invariant under translations 
(i.e.\ $\{d^a, P_b\} = 0$), transforms as a vector under 
spatial rotations (i.e.\ $\{J_{ab}, d^c\} = \delta_a^c d_b - \delta_b^c d_a$) 
and is invariant under time reversal with respect to 
the hyperplane through the origin and orthogonal to 
$u$ (i.e.\ $\vec d \circ T_u = \vec d$).

\begin{lem} \label{lem:invar_vector_no_mom_comp}
	Let $\vec A$ be a $\mathbb R^3$-valued $C^1$ phase 
	space function on a classical elementary system with 
	timelike four-momentum that is invariant under 
	translations, transforms as a vector under spatial 
	rotations and is invariant under time reversal. Then 
	$\vec A \cdot \vec P = 0$.

	\begin{proof}
		Since $\vec P$ is invariant under translations and 
		a vector under rotations, $\vec A \cdot \vec P$ is 
		invariant under translations and rotations. By lemma 
		\ref{lem:invar_function}, $\left. \vec A \cdot \vec P \right|_{\Gamma^*}$ 
		is a function of $|\vec P|, \hat P \cdot \vec s$. 
		This means we have
		\begin{equation} \label{eq:pf_lem_scalar_product_mom_1}
			\left. \vec A \cdot \vec P \right|_{\Gamma^*} = F(|\vec P|, \hat P \cdot \vec s)
		\end{equation}
		for some function $F\colon \mathbb R_+ \times [-S,S] \to \mathbb R$.

		Now considering time reversal $T_u$, on the one hand 
		we have (using lemma \ref{lem:time_rev})
		\begin{subequations}
		\begin{equation}
			|\vec P| \circ T_u = |\vec P \circ T_u| = |-\vec P| = |\vec P|
		\end{equation}
		and
		\begin{align}
			(\hat P \cdot \vec s) \circ T_u &= \left(\frac{1}{2} {^{(3)}\varepsilon^{abc}} \hat P_a J_{bc}\right) \circ T_u \nonumber\\
			&= \frac{1}{2} {^{(3)}\varepsilon^{abc}} (\hat P_a \circ T_u) (J_{bc} \circ T_u) \nonumber\\
			&= \frac{1}{2} {^{(3)}\varepsilon^{abc}} (-\hat P_a) (-J_{bc}) \nonumber\\
			&= \frac{1}{2} {^{(3)}\varepsilon^{abc}} \hat P_a J_{bc} \nonumber\\
			&= \hat P \cdot \vec s,
		\end{align}
		\end{subequations}
		implying
		\begin{equation} \label{eq:pf_lem_scalar_product_mom_2}
			F(|\vec P|, \hat P \cdot \vec s) \circ T_u = F(|\vec P| \circ T_u, (\hat P \cdot \vec s) \circ T_u) = F(|\vec P|, \hat P \cdot \vec s).
		\end{equation}
		On the other hand, $\vec A$ is invariant under time 
		reversal while $\vec P$ changes its sign, implying that 
		$(\vec A \cdot \vec P) \circ T_u = - \vec A \cdot \vec P$. 
		Combining this with \eqref{eq:pf_lem_scalar_product_mom_1} 
		and \eqref{eq:pf_lem_scalar_product_mom_2}, we obtain 
		$\left. \vec A \cdot \vec P \right|_{\Gamma^*} = 0$, 
		and continuity implies $\vec A \cdot \vec P = 0$.
	\end{proof}
\end{lem}

For zero spin, we can easily complete the proof of the 
Newton--Wigner theorem. Since the difference vector 
$\vec d$ is translation invariant and the $P_a$ form a 
complete involutive set on $\Gamma$, $\vec d$ must be a 
function of $\vec P$. Then since it is a vector under 
rotations, it must be of the form
\begin{equation}
	\vec d(\vec P) = F(|\vec P|) \vec P
\end{equation}
for some function $F$ of $|\vec P|$. Then, since 
according to lemma \ref{lem:invar_vector_no_mom_comp} 
$\vec d \cdot \vec P$ is zero, $\vec d$ is zero. Thus, 
for the spin-zero case, we have proved the Newton--Wigner 
theorem without any use of the condition of 
Poisson-commuting components of the position observable.

For the non-zero spin case, we continue as follows.

\begin{lem} \label{lem:invar_vector_in_basis}
	Let $\vec A$ be a $\mathbb R^3$-valued $C^1$ phase space 
	function on a classical elementary system with timelike 
	four-momentum and non-zero spin that is invariant under 
	translations, transforms as a vector under spatial 
	rotations and satisfies $\vec A \cdot \vec P = 0$. 
	Then it is of the form
	\begin{equation} \label{eq:invar_vector_in_basis}
		\vec A = B \hat P \times \vec s + C \hat P \times (\hat P \times \vec s)
	\end{equation}
	on $\Gamma^* \setminus \{\vec s \parallel \hat P\}$, 
	where $B$ and $C$ are $C^1$ functions of $|\vec P|$ 
	and $\hat P \cdot \vec s$, i.e.\ $C^1$ functions
	\[B,C \colon \mathbb R_+ \times (-S,S) \to \mathbb R.\]

	\begin{proof}
		For the whole of this proof, we will work on 
		$\tilde\Gamma := \Gamma^* \setminus \{\vec s \parallel \hat P\}$. 
		Since evaluated at each point of $\tilde\Gamma$, the 
		$\mathbb R^3$-valued functions 
		$\hat P, \hat P \times \vec s, \hat P \times (\hat P \times \vec s)$ 
		form an orthogonal basis of $\mathbb R^3$, and since 
		we have $\vec A \cdot \vec P = 0$, we can write 
		$\vec A$ in the form \eqref{eq:invar_vector_in_basis} with 
		coefficients $B,C$ given by
		\begin{align}
			B &= \frac{\vec A \cdot (\hat P \times \vec s)}{|\hat P \times \vec s|} \; , \\
			C &= \frac{\vec A \cdot (\hat P \times (\hat P \times \vec s))}{|\hat P \times (\hat P \times \vec s)|} \; .
		\end{align}
		Since $\vec A$, $\vec P$ and $\vec s$ are invariant 
		under translations and vectors under rotations, these 
		equations imply that $B,C$ are invariant under 
		translations and rotations. The result follows with 
		lemma \ref{lem:invar_function}.
	\end{proof}
\end{lem}

Now we consider again the difference vector 
$\vec d = \vec X - \vec x$. It satisfies 
$\vec d \cdot \vec P = 0$ by lemma \ref{lem:invar_vector_no_mom_comp}, 
and thus we have
\begin{equation}
	\vec X \cdot \vec P = \vec x \cdot \vec P.
\end{equation}
Since we assume that the components of $\vec X$ 
Poisson-commute with each other and that 
$\{X^a, P_b\} = \delta^a_b$, this implies
\begin{equation}
	\{X^a, \vec x \cdot \vec P\} = \{X^a, \vec X \cdot \vec P\} = X^a.
\end{equation}
Combining this with $\{x^a, \vec x \cdot \vec P\} = x^a$, 
we obtain
\begin{equation}
	\{d^a, \vec x \cdot \vec P\} = d^a.
\end{equation}
On the other hand, for any function $F$ of $\vec P$ and 
$\vec s$, we have
\begin{equation}
	\{F(\vec P, \vec s), \vec x \cdot \vec P\} = \{F(\vec P, \vec s), x^a\} P_a = - \frac{\partial F(\vec P, \vec s)}{\partial P_a} P_a = - |\vec P| \left. \frac{\partial F}{\partial |\vec P|} \right|_{\hat P = \mathrm{const.}, \vec s = \mathrm{const.}}.
\end{equation}
This implies
\begin{equation} \label{eq:der_diff}
	\vec d = - |\vec P| \left. \frac{\partial \vec d}{\partial |\vec P|} \right|_{\hat P = \mathrm{const.}, \vec s = \mathrm{const.}}.
\end{equation}
Combining lemmas \ref{lem:invar_vector_no_mom_comp} and 
\ref{lem:invar_vector_in_basis}, we know that $\vec d$ 
has the form \eqref{eq:invar_vector_in_basis} on 
$\Gamma^* \setminus \{\vec s \parallel \hat P\}$ for two 
functions $B, C \colon \mathbb R_+ \times (-S,S) \to \mathbb R$. 
Thus \eqref{eq:der_diff} implies the two equations
\begin{equation}
	B(|\vec P|, \hat P \cdot \vec s) = - |\vec P| \frac{\partial B(|\vec P|, \hat P \cdot \vec s)}{\partial |\vec P|}, \; C(|\vec P|, \hat P \cdot \vec s) = - |\vec P| \frac{\partial C(|\vec P|, \hat P \cdot \vec s)}{\partial |\vec P|}
\end{equation}
on $\mathbb R_+ \times (-S,S)$. These equations 
determine the $|\vec P|$ dependence of $B$ and $C$; 
they must be proportional to $|\vec P|^{-1}$. However, 
for $\vec d$ to be $C^1$ on the whole of $\Gamma$, in 
fact for \eqref{eq:invar_vector_in_basis} not to diverge 
as $|\vec P|\to 0$ even when coming from a \emph{single} 
direction $\hat P$, we then need $B$ and $C$ to vanish. 
Continuity implies $\vec d = 0$ on all of $\Gamma$. This 
finishes the proof of the Newton--Wigner theorem.
\end{proof}

\section{Conclusion}

In this paper we have studied the localisation 
problem for classical system whose phase space 
is a symplectic manifold. We focussed on the 
Newton--Wigner position observable and asked for precise 
characterisations of it in order to gain 
additional understanding, over and above that 
already known from its practical use for the 
solution of concrete problems of motion, e.g.,
in astrophysics \cite{Steinhoff:2011,Schaefer.Jaranowski:2018}. 
We proved two theorems that we believe advance our 
understanding in the desired direction: 
first we showed how Fleming's geometric scheme 
\cite{Fleming:1965a} in combination with the 
characterisation of worldlines through SSCs (Spin 
Supplementary Conditions) allows to give a 
precise meaning to, and proof of, the fact 
that the Newton--Wigner position is the unique 
centre of spin. Given that interpretation, it 
also offers an insight as to why the Newton--Wigner 
SSC uses a somewhat unnatural looking `hybrid' 
combination $f = u + \frac{P}{mc}$, where $u$ is 
`external' or `kinematical', and $P$ is `internal' 
or `dynamical'. Then, restricting to elementary 
systems, i.e.\ systems whose phase space admits a 
transitive action of the proper orthochronous 
Poincaré group, we proved again a uniqueness result 
to the effect that the Newton--Wigner observable 
is the unique phase space function whose 
components satisfy the `familiar' Poisson 
relations, provided it is continuously 
differentiable, time-reversal invariant, and 
transforms as a vector under spatial rotations. 
These properties seem to be the underlying reason 
for the distinguished rôle it plays in 
solution strategies like those of 
\cite{Steinhoff:2011,Schaefer.Jaranowski:2018}\footnote
{Despite the fact that on a more general level of 
theorisation other choices (characterised by other 
SSCs) are often considered more appropriate; 
see, e.g., \cite{Puetzfeld.EtAl:2015}.}, 
which in recent years have lead to astounding progress in 
the Hamiltonian analytical understanding of the dynamics 
of binary systems of spinning compact objects: the calculations 
have been pushed to ever higher post-Newtonian orders, 
starting from the next-to-leading order for spin--orbit 
and spin--spin effects, i.e.\ order $c^{-4}$, in 
\cite{Steinhoff.Hergt.Schaefer:2008a,Steinhoff.Schaefer.Hergt:2008,
Steinhoff.Hergt.Schaefer:2008b}, 
and most recently reaching a complete description 
at `4.5th' post-Newtonian order, i.e.\ order $c^{-9}$, 
in \cite{Levi.Mougiakakos.Vieira:2019,Antonelli.EtAl:2020}.
We believe 
that our results add a conceptually clear and 
mathematically precise Hamiltonian underpinning 
of what the choice of the Newton--Wigner observable 
entails, at least in a special-relativistic 
context or, more generally, in general-relativistic 
perturbation theory around Minkowski space.

\section*{Acknowledgements}
This work was supported by the Deutsche 
Forschungsgemeinschaft through the Collaborative 
Research Centre 1227 (DQ-mat), projects B08/A05.

\nocite{apsrev41Control}
\bibliography{NW_class,revtex-custom}

\pagebreak
\appendix

\section{Sign conventions for generators of special orthogonal groups}
\label{app:sign_convention_so}

Let $V$ be a finite-dimensional real vector space 
with a non-degenerate, symmetric bilinear form 
$g\colon V \times V \to \mathbb R$. Note that we do 
not assume anything about the signature of $g$. 
We introduce the `musical isomorphism'
\begin{equation} \label{eq:musical_isom}
	V \to V^*, v \mapsto v^\flat := g(v, \cdot)
\end{equation}
induced by $g$.

We fix a basis $\{e_a\}_a$ of $V$. As bases for its 
dual vector space $V^*$ we distinguish its natural 
dual basis $\{\theta^a\}_a$, where $\theta^a(e_b)=\delta^a_b$,
and the ($g$-dependent) image of $\{e_a\}_a$ under 
\eqref{eq:musical_isom}, which is just $\{e^\flat_a\}_a$, 
where $e^\flat_a=g_{ab}\theta^b$, so that $e^\flat_a(e_b)=g_{ab}$.
The reason for this will become clear now. 

For each $a,b \in \{1, \dots, \dim V\}$ we introduce 
the endomorphism
\begin{equation} \label{eq:lie_basis-1}
	B_{ab} := e_a \otimes e_b^\flat - e_b \otimes e_a^\flat \in \mathrm{End}(V),
\end{equation}
which satisfies 
\begin{equation}
	g(v, B_{ab}(w)) = g(v,e_a) g(e_b,w) - g(v,e_b) g(e_a,w) = -g(B_{ab}(v), w).
\end{equation}
This means that $B_{ab}$ is anti-self-adjoint 
with respect to $g$ and hence that it is an 
element of the Lie algebra $\mathfrak{so}(V,g)$ 
of the Lie group $\mathsf{SO}(V,g)$ of special 
orthogonal transformations of $(V,g)$:
\begin{equation}
	B_{ab} \in \mathfrak{so}(V,g).
\end{equation}
As $B_{ab} = -B_{ba}$, it is the set 
$\{B_{ab}: 1\le a < b\le \dim V\}$ which is linearly 
independent and of the same dimension as 
$\mathfrak{so}(V,g)$. Hence this set forms a basis of 
$\mathfrak{so}(V,g)$ so that any 
$\omega \in \mathfrak{so}(V,g)$ can be uniquely 
written in the form
\begin{equation}
	\omega = \sum_{1\le a < b \le \dim V} \omega^{ab} B_{ab} = \frac{1}{2} \omega^{ab} B_{ab} \; ,
\end{equation}
where 
\begin{equation} \label{eq:lie_basis_anti_self_adj}
	\omega^{ab} = -\omega^{ba}\,.
\end{equation}

This representation can easily be compared 
to the usual one in terms of the metric-independent basis 
$\{e_a\otimes\theta^b : 1\le a,b \le \dim V\}$ of 
$\mathrm{End}(V)$ in the following way: for 
$\omega=\omega^a_{\hphantom{a}c}\, e_a \otimes \theta^c$,
we have $\omega\in \mathfrak{so}(V,g)$ if and only if 
\begin{equation} \label{eq:end_basis_anti_self_adj}
	\omega^a_{\hphantom{a}c} \, g^{cb} = - \omega^b_{\hphantom{a}c} \, g^{ca} \; .
\end{equation}
It is the obvious simplicity of \eqref{eq:lie_basis_anti_self_adj}
as opposed to \eqref{eq:end_basis_anti_self_adj}
as conditions for $\omega\in\mathrm{End}(V)$
being contained in  
$\mathfrak{so}(V,g)\subset\mathrm{End}(V)$ that 
makes it easier to work with the basis 
$e_a\otimes e^\flat_b$ of $\mathrm{End}(V)$ rather 
than $e_a\otimes\theta^b$. Note that the components 
of $\omega$ with respect to the two bases considered 
above are connected by the equation
\begin{equation}
	\omega^{ab} = \omega^a_{\hphantom{a}c} \, g^{cb} \; .
\end{equation}

The basis elements $B_{ab}$ satisfy 
the commutation relations
\begin{align}
	[B_{ab}, B_{cd}] &= g_{bc} B_{ad} + g_{ad} B_{bc} - g_{ac} B_{bd} - g_{bd} B_{ac} \nonumber\\
	&= g_{bc} B_{ad} + \text{(antisymm.)},
\end{align}
where `antisymm.' is as explained below 
equation \eqref{eq:Poinc_gen}.

From now on, we will assume the basis $\{e_a\}_a$ to be 
orthonormal. For notational convenience, 
for $a,b\in \{1, \dots, \dim V\}$ we define
\begin{equation}
	\varepsilon_{ab} := g_{aa} g_{bb} = \pm 1
\end{equation}
which has the value $+1$ if $g_{aa} = g(e_a, e_a)$ 
and $g_{bb} = g(e_b,e_b)$ have the same sign, 
and $-1$ if they have opposite signs\footnote
	{Note that repeated indices on the same 
	level, i.e.\ both up or both down, are not to be 
	summed over.}.

We now want to compute the exponential  
$\exp(\alpha B_{ab}) \in \mathsf{SO}(V,g)$. 
At first, we note that
\begin{align}
	(B_{ab})^2 
	&= - g_{bb} e_a \otimes e_a^\flat 
	 - g_{aa} e_b \otimes e_b^\flat \nonumber\\
	&= - \varepsilon_{ab} \; \mathrm{Pr}_{ab}\,,
\end{align}
where $\mathrm{Pr}_{ab} := \mathrm{Pr}_{\mathrm{span}\{e_a, e_b\}}$ 
denotes the $g$-orthogonal projector onto the plane 
$\mathrm{span}\{e_a, e_b\}$ in $V$.\footnote
	{In the general case of two linearly independent 
	vectors $v,w \in V$, not necessarily orthonormal, the 
	orthogonal projector is given by
		\begin{equation}
			\mathrm{Pr}_{\mathrm{span}(v,w) = \frac{1}{g(v,v) g(w,w) - (g(v,w))^2}} \left[ g(w,w) \; v \otimes v^\flat + g(v,v) \; w \otimes w^\flat - g(v,w) \; (v \otimes w^\flat + w \otimes v^\flat) \right],
		\end{equation}
		implying
		\begin{align}
			(v \otimes w^\flat - w \otimes v^\flat)^2 &= - g(w,w) \; v \otimes v^\flat - g(v,v) \; w \otimes w^\flat + g(v,w) \; (v \otimes w^\flat + w \otimes v^\flat) \nonumber\\
			&= - \left[ g(v,v) g(w,w) - (g(v,w))^2 \right] \mathrm{Pr}_{\mathrm{span}(v,w)}.
		\end{align}
	}
Using this and $B_{ab} \circ \mathrm{Pr}_{ab} = B_{ab}$, 
the exponential series evaluates to
\begin{align}
	\exp(\alpha B_{ab}) 
	&= (\mathrm{id}_V - \mathrm{Pr}_{ab}) 
	 + \sum_{k=0}^\infty \frac{1}{(2k)!} \, 
	   \alpha^{2k} (-\varepsilon_{ab})^k \, 
	   \mathrm{Pr}_{ab} \nonumber\\ &\quad+ 
	   \sum_{k=0}^\infty \frac{1}{(2k+1)!} \, 
	   \alpha^{2k+1} (-\varepsilon_{ab})^k \, 
	    B_{ab} \circ \mathrm{Pr}_{ab} \nonumber\\
	&= (\mathrm{id}_V - \mathrm{Pr}_{ab})
	 + \left\{ \!\! \begin{aligned}
		&\cos(\alpha) \, \mathrm{id}_V + \sin(\alpha) \, 
		 B_{ab} \, , & \varepsilon_{ab} = +1 \\
		&\cosh(\alpha) \, \mathrm{id}_V + \sinh(\alpha) \, 
		 B_{ab} \, , & \varepsilon_{ab} = -1
	\end{aligned} \right\} \circ \mathrm{Pr}_{ab} \; .
\end{align}
Geometrically, this transformation is either a 
rotation by angle $\alpha$ (for $\varepsilon_{ab} = +1$) 
or a boost by rapidity $\alpha$ (for 
$\varepsilon_{ab} = -1$) in the plane 
$\mathrm{span}\{e_a, e_b\}$. The direction of the 
transformation depends on the signs of $g_{aa}, g_{bb}$:
\begin{itemize}
	\item $\varepsilon_{ab} = +1$:
		\begin{enumerate}[label=(\roman*)]
			\item $g_{aa} = g_{bb} = +1$: We have 
				$B_{ab}(e_a) = -e_b, B_{ab}(e_b) = e_a$. 
				Thus, $\exp(\alpha B_{ab})$ is a rotation 
				by $\alpha$ \emph{from $e_b$ towards $e_a$}.
			\item $g_{aa} = g_{bb} = -1$: We have 
				$B_{ab}(e_a) = e_b, B_{ab}(e_b) = -e_a$. 
				Thus, $\exp(\alpha B_{ab})$ is a rotation 
				by $\alpha$ \emph{from $e_a$ towards $e_b$}.
		\end{enumerate}
	\item $\varepsilon_{ab} = -1$:
		\begin{enumerate}[label=(\roman*)]
			\item $g_{aa} = +1, g_{bb} = -1$: We have 
				$B_{ab}(e_a) = -e_b, B_{ab}(e_b) = -e_a$. 
				Thus, $\exp(\alpha B_{ab})$ is a boost 
				by $\alpha$ \emph{`away' from $e_a + e_b$}.
			\item $g_{aa} = -1, g_{bb} = +1$: We have 
				$B_{ab}(e_a) = e_b, B_{ab}(e_b) = e_a$. 
				Thus, $\exp(\alpha B_{ab})$ is a boost 
				by $\alpha$ \emph{`towards' $e_a + e_b$}.
		\end{enumerate}
\end{itemize}

Now we will apply the preceding considerations to the case 
of (the `difference' vector space of) Minkowski spacetime, 
where for now \emph{we leave open the signature convention 
for the metric} (either $(+{--}-)$ or $(-{++}+)$). We work 
with respect to a positively oriented orthonormal basis 
$\{e_\mu\}_{\mu = 0, \dots, 3}$ where $e_0$ is timelike. 
Latin indices will denote spacelike directions.

In the case of `mostly minus' signature $(+{--}-)$, 
$B_{ab}$ generates rotations from $e_a$ towards $e_b$ and 
$B_{a0}$ generates boosts (with respect to $e_0$) in 
direction of $e_a$. In the case of `mostly plus' signature 
$(-{++}+)$, $B_{ba} = - B_{ab}$ generates rotations from 
$e_a$ towards $e_b$ and $B_{0a} = - B_{a0}$ generates 
boosts (with respect to $e_0$) in direction of $e_a$.

Thus, since we want to use the notation $J_{ab}$ for the 
spacelike rotational generator generating rotations from 
$e_a$ towards $e_b$, we have to set
\begin{equation} \label{eq:sign_convention_so}
	J_{\mu\nu} = \begin{cases}
		B_{\mu\nu} & \text{for $(+{--}-)$ signature},\\
		-B_{\mu\nu} & \text{for $(-{++}+)$ signature}
	\end{cases}
\end{equation}
for the Lorentz generators. Adopting this convention, 
boosts in direction of $e_a$ are then generated 
by $J_{a0}$. The commutation relations for the 
$J_{\mu\nu}$ are
\begin{equation}
	[J_{\mu\nu}, J_{\rho\sigma}] = \begin{cases}
		\eta_{\mu\sigma} J_{\nu\rho} + \text{(antisymm.)} & \text{for $(+{--}-)$ signature},\\
		\eta_{\mu\rho} J_{\nu\sigma} + \text{(antisymm.)} & \text{for $(-{++}+)$ signature},
	\end{cases}
\end{equation}
and general Lorentz algebra elements 
$\omega \in \mathrm{Lie}(\mathcal L)$ can be written as
\begin{equation}
	\omega = \pm \frac{1}{2} \omega^{\mu\nu} J_{\mu\nu} \; \text{with} \; \omega^{\mu\nu} = \omega^\mu_{\hphantom{\mu}\rho} \; \eta^{\rho\nu}
\end{equation}
in terms of their components $\omega^\mu_{\hphantom{\mu}\rho}$ 
as endomorphisms, where the upper/lower sign holds for 
$(+{--}-)$/$(-{++}+)$ signature.

\section{Notes on the adjoint representation}
\label{app:adj_rep}

Here we wish to make a few remarks and collect a 
few formulae concerning the adjoint and co-adjoint 
representation, which will be made use of in the main 
text. 

In the defining representation on $V$, an element 
$\Lambda \in \mathsf{GL}(V)$ is given in terms of the 
basis $\{e_a\}_a$ by the coefficients 
$\Lambda^a_{\phantom{a}_b}$, where
\begin{equation} \label{eq:coeff_def_rep}
	\Lambda e_a = \Lambda^b_{\phantom{b}a} \, e_b \; .
\end{equation}  
This defines a left action of $\mathsf{GL}(V)$ on 
$V$. The corresponding left action of 
$\mathsf{GL}(V)$ on the dual space $V^*$ is given 
by the inverse-transposed, i.e.\ 
$\mathsf{GL}(V) \times V^* \to V^*$, 
$(\Lambda,\alpha) \mapsto (\Lambda^{-1})^\top
\alpha := \alpha \circ \Lambda^{-1}$. For the  
basis $\{\theta^a\}_a$ of $V^*$ dual to the 
basis $\{e_a\}_a$ this means
\begin{equation}
	\theta^a\circ \Lambda^{-1} = (\Lambda^{-1})^a_{\phantom{a}b} \, \theta^b \; .
\end{equation}  
In contrast, for the basis $\{e_a^\flat\}_a$
of $V^*$, this reads in general
\begin{equation}
	e_b^\flat \circ \Lambda^{-1} = g^{ac} g_{bd} (\Lambda^{-1})^d_{\phantom{d}c} \, e_a^\flat \; ,
\end{equation}  
which for isometries $\Lambda \in \mathsf{O}(V,g)$ simply becomes
\begin{equation} \label{eq:coeff_def_rep_dual_isom}
	e^\flat_b\circ \Lambda^{-1} = \Lambda^a_{\phantom{a}b} \, e^\flat_a \; .
\end{equation}  

The adjoint representation of $\mathsf{GL}(V)$ on 
$\mathrm{End}(V)\cong V\otimes V^*$ or any Lie subalgebra
of $\mathrm{End}(V)$ is by conjugation, which for 
our basis \eqref{eq:lie_basis-1} implies, using \eqref{eq:coeff_def_rep} and \eqref{eq:coeff_def_rep_dual_isom},
\begin{equation}
	\mathrm{Ad}_\Lambda B_{ab}
	= \Lambda \circ B_{ab} \circ \Lambda^{-1}
	= \Lambda^c_{\phantom{c}a} \Lambda^d_{\phantom{d}b} \, B_{cd}
	\quad \text{for} \; \Lambda \in \mathsf{O}(V,g).
\end{equation}

The adjoint 
representation of the inhomogeneous 
group $\mathsf{GL}(V)\ltimes V$ on its Lie algebra 
$\mathrm{End}(V)\oplus V$ is given by, for any 
$X \in \mathrm{End}(V)$ and $y \in V$,
\begin{equation}
	\mathrm{Ad}_{(\Lambda,a)}(X,y) =
	\left( \Lambda\circ X\circ\Lambda^{-1} ,
	\Lambda y - (\Lambda \circ X \circ \Lambda^{-1}) a \right).
\end{equation}
In the main text we will use this formula for 
$(\Lambda,a)$ being replaced by its inverse 
$(\Lambda,a)^{-1} = (\Lambda^{-1}, -\Lambda^{-1}a)$:
\begin{equation}
	\mathrm{Ad}_{(\Lambda,a)^{-1}}(X,y) =
	\left( \Lambda^{-1} \circ X \circ \Lambda ,
	\Lambda^{-1} y + (\Lambda^{-1} \circ X) a \right)
\end{equation}  
Applied to the basis vectors separately, 
i.e.\ to $(X,y) = (0,e_b)$ and $(X,y) = (B_{bc},0)$,
for $\Lambda \in \mathsf{O}(V,g)$ we get 
\begin{subequations} \label{eq:ad_rep}
\begin{alignat}{2}
	&\mathrm{Ad}_{(\Lambda,a)^{-1}}(0,e_b)
	&&= \left( 0, (\Lambda^{-1})^c_{\phantom{c}b} e_c \right) \\
	&\mathrm{Ad}_{(\Lambda,a)^{-1}}(B_{bc},0)
	&&= \left( (\Lambda^{-1})^d_{\phantom{d}b} (\Lambda^{-1})^e_{\phantom{e}c} \, B_{de} ,
	- a_b (\Lambda^{-1})^d_{\phantom{d}c} e_d
	+ a_c (\Lambda^{-1})^d_{\phantom{d}b} e_d \right)
\end{alignat}  
\end{subequations}
where $a_b := e_b^\flat(a) = g_{bc} a^c$ in the second 
equation. From these equations we immediately 
deduce \eqref{eq:coadjoint_rep_components} in the 
case of four spacetime dimensions (greek indices) 
and signature mostly plus, in which case 
$J_{\mu\nu} = -B_{\mu\nu}$ according to
\eqref{eq:sign_convention_so}.

\end{document}